\documentclass[preprint]{autart}

\usepackage{amssymb}
\usepackage{amsmath,amsfonts}
\usepackage{mathtools}
\usepackage{dsfont}
\usepackage{pifont}
\usepackage{array}
\usepackage{booktabs}
\usepackage{textcomp}
\usepackage{stfloats}
\usepackage{url}
\usepackage{verbatim}
\usepackage{graphicx}

\usepackage[round]{natbib}

\usepackage[hidelinks]{hyperref}

\usepackage{subcaption}

\usepackage{algorithmic}
\usepackage{algorithm}

\usepackage{xparse}
\usepackage{etoolbox} 

\usepackage{times}
\newtheorem{problem}{Problem}
\newtheorem{assumption}{Assumption}

\usepackage{setspace}
\usepackage{tikz}
\usepackage{flowchart}
\usetikzlibrary{arrows}
\usepackage{textcomp}
\usetikzlibrary{shapes,arrows,positioning,calc}
\tikzstyle{startstop} = [rectangle, rounded corners, minimum width=2cm, minimum height=1cm,text centered, draw=black]
\tikzstyle{process} = [rectangle, minimum width=2cm, minimum height=1cm, text centered, draw=black]
\tikzstyle{decision} = [diamond, minimum width=0.5cm, minimum height=0.5cm, text centered, draw=black]
\tikzstyle{arrow} = [thick,->,>=stealth]
\definecolor{lightgray}{gray}{0.85}
\definecolor{lightergray}{gray}{0.9}

\definecolor{darkgreen}{RGB}{34,139,34}   
\definecolor{darkred}{RGB}{178,34,34}     

\usepackage[framemethod=TikZ]{mdframed}
\mdfdefinestyle{callout}{%
	linecolor=black,
	linewidth=1pt,%
	roundcorner=0pt,
	innertopmargin=4pt,
	innerbottommargin=4pt,
	innerrightmargin=5pt,
	innerleftmargin=5pt,
	leftmargin = 0pt,
	rightmargin = 0pt,
	backgroundcolor=lightergray
}

\usepackage{dsfont}

\usepackage{dirtytalk}

\usepackage{endnotes}

\definecolor{cadetblue}{rgb}{0.33, 0.41, 0.58}

\DeclareMathOperator{\argmin}{\mathrm{argmin}}

\DeclareMathOperator{\EV}{\mathds{E}}

\DeclareMathOperator{\tr}{\mathrm{tr}}
\DeclareMathOperator{\col}{\mathrm{col}}
\DeclareMathOperator{\diag}{\mathrm{diag}}

\DeclareMathAlphabet{\doublestruck}{U}{BOONDOX-ds}{m}{n}
\newcommand{\ones}[0]{\mathds{1}}
\newcommand{\zeros}[0]{\doublestruck{0}}

\newcommand{\bepsilon}[0]{\boldsymbol{\epsilon}}

\newcommand{\bomega}[0]{\boldsymbol{\omega}}

\newcommand{\N}[0]{\mathbb{N}}
\newcommand{\Nz}[0]{\mathbb{N}_0}

\newcommand{\R}[0]{\mathbb{R}}

\newcommand{\Rnn}[0]{\mathbb{R}_{\geq0}}

\newcommand{\Ecal}[0]{\mathcal{E}}

\newcommand{\Gcal}[0]{\mathcal{G}}

\newcommand{\Ical}[0]{\mathcal{I}}

\newcommand{\Ocal}[0]{\mathcal{O}}
\newcommand{\Pcal}[0]{\mathcal{P}}

\newcommand{\Scal}[0]{\mathcal{S}}
\newcommand{\Tcal}[0]{\mathcal{T}}

\newcommand{\Vcal}[0]{\mathcal{V}}

\usepackage{xstring}

\usepackage{setspace}

\usepackage{tikz}
\usetikzlibrary{shapes,arrows}
\usetikzlibrary{arrows,calc,positioning}

\tikzset{
	block/.style = {draw, rectangle,
		minimum height=1.2cm,
		minimum width=1.2cm},
	input/.style = {coordinate,node distance=1cm},
	output/.style = {coordinate,node distance=2cm},
	arrow/.style={draw, -latex,node distance=2cm},
	pinstyle/.style = {pin edge={latex-, black,node distance=1cm}},
	sum/.style = {draw, circle, node distance=1cm},
}

\makeatletter
\patchcmd{\AU@maketitle}{\vskip 12pt}{\vskip 4pt}{}{}
\makeatother

\makeatletter
\renewcommand{\@authorsize}{\normalsize}
\makeatother

\DeclareMathSizes{10}{9}{7}{7}

\newtheorem{lemma}{Lemma}
\newtheorem{definition}{Definition}
\newtheorem{remark}{Remark}

\makeatletter
\long\def\@makecaption#1#2{%
	\vskip\abovecaptionskip
	\sbox\@tempboxa{#1: #2}%
	\ifdim \wd\@tempboxa >\hsize
	#1: #2\par
	\else
	\global \@minipagefalse
	\hbox to\hsize{\hfil\box\@tempboxa\hfil}%
	\fi
	\vskip\belowcaptionskip}
\makeatother

\begin{document}

\begin{frontmatter}


\title{Consistent Distributed Cooperative Localization for\\Ultra Large-Scale Multi-agent Systems}\vspace{-0.5cm}



\author[tue,ist]{Leonardo Pedroso}\ead{l.pedroso@tue.nl},    
\author[tue]{W.P.M.H. (Maurice) Heemels}\ead{m.heemels@tue.nl},               
\author[ist]{Pedro Batista}\ead{pbatista@isr.tecnico.ulisboa.pt}  

\address[tue]{Control Systems Technology section, Department of Mechanical Engineering, Eindhoven University of Technology, The Netherlands\vspace{-0.1cm}}  
\address[ist]{Institute for Systems and Robotics, Instituto Superior T\'ecnico, Universidade de Lisboa, Portugal}             

	\makeatletter
\def\AU@fld@textfont{\small}
\makeatother

\begin{keyword}                           
	Cooperative localization, Distributed estimation, Covariance intersection, Multi-agent systems, Ultra large-scale systems, Information fusion
\end{keyword}          

\begin{abstract}   
Cooperative localization (CL) is fundamental in emerging multi-agent systems, where agents fuse local sensing data with exchanged information to estimate their own states. At a large scale, however, tracking cross-correlations becomes infeasible, preventing the use of optimal filters.
Ignoring or underestimating these correlations leads to overconfident, and thus inconsistent,  estimates. Existing CL algorithms achieve good performance and consistency typically at the expense of communication, computation, or memory that scales with the network size. This is incompatible with ultra large-scale systems (ULSS)—for example, satellite mega-constellations—where per-agent resources are limited and must remain independent of the number of agents. 
This reveals a critical gap: no existing CL method is simultaneously \emph{well-performing}, \emph{consistent}, and \emph{ULSS-scalable}.
This paper introduces a new CL framework that addresses this gap using the recently proposed overlapping covariance intersection methodology, which enables agents to exploit limited structural information about cross-correlations without compromising consistency.
The resulting CL algorithm leads to optimal conservative covariance propagation using only locally available information. The method is fully distributed, scalable to an ultra large scale, and provably recursively consistent. 
Simulations demonstrate substantial performance improvement over state-of-the-art consistent CL approaches while preserving scalability.
\vspace{-0.3cm}
\end{abstract}

\end{frontmatter}

\vspace{-0.8cm}

\section{Introduction}\label{sec:introduction}

\subsection{Motivation} \vspace{-0.2cm}
Cooperative localization (CL) is a cornerstone problem in the coordination of large-scale multi-agent systems. It concerns the task where each agent estimates its own state by combining (i)~local measurements—such as odometry or relative measurements with respect to neighbors—and (ii)~information exchanged through communication. This paradigm enables accurate and robust localization in scenarios where global positioning infrastructure is unavailable, unreliable, or deliberately absent. Consequently, CL has become critical for applications ranging from ground robot teams \citep{FoxBurgardEtAl2000,RoumeliotisBekey2002}, to formations of unmanned aerial vehicles \citep{BeregDiaz-BanezEtAl2015,ThienKimYoonsoo2018}, cooperative underwater vehicles \citep{ViegasBatistaEtAl2012,YuanLichtEtAl2018}, and satellite formations \citep{RussellCarpenter2002,IvanovMonakhovaEtAl2019}. These approaches embrace the distributed nature of multi-agent systems, offering resilience, scalability, and flexibility \citep{Sahin2005}. Despite substantial progress, distributed CL remains challenging. A central difficulty lies in properly accounting for correlations between agents’ estimates. On the one hand, if correlations are ignored or underestimated, the estimator becomes overconfident and it is said to be \emph{inconsistent}, i.e., the estimated covariance of the fused estimate is smaller than the true covariance \citep{HowardMataricEtAl2003,PanzieriPascucciSetola2006}. Estimator inconsistency can have dire consequences, e.g., in debris avoidance maneuvers in satellite formations. On the other hand, if correlations are treated too conservatively, estimation is \emph{consistent} but performance is typically significantly degraded. This fundamental tension between consistency and performance is at the heart of the CL problem.

Recently, the emergence of \emph{ultra large-scale systems} (ULSS)—such as low Earth orbit mega-constellations \citep{DaehnickKlinghofferMaritzEtAl2020,XieZhanEtAl2021}—has brought CL challenges to a new level \citep[Section~5]{PedrosoBatistaEtAl2025ULSS}. Indeed, most well-performing CL solutions rely on algorithms whose communication, computation, or memory requirements per agent grow with the size of the entire network. Beyond a certain scale, such procedures inevitably become infeasible. The term ULSS was introduced in the computational engineering field by \citet{FeilerGabrielGoodenoughEtAl2006} and, more recently, in the context of control systems by \citet{PedrosoBatistaEtAl2025ULSS} to formally characterize this regime. There, it is argued that the transition to the ultra large scale (ULS) calls for additional requirements beyond standard information-structure constraints to ensure real-world implementability. In particular, to be scalable to ULSS, both working and design phases of CL control algorithms must be entirely distributed across agents and their per-agent communication, computation, and memory usage must not grow with the number of agents. These disruptive constraints make ULSS fundamentally different from classical large-scale systems and demand new algorithmic paradigms. Motivated by this paradigm shift, this paper focuses on \emph{CL algorithms on the ultra large scale}.

\subsection{State-of-the-art}\vspace{-0.2cm}
Several families of CL solutions have been proposed. First, centralized-equivalent schemes—such as junction-tree protocols \citep{DaiMuWu2016} or measurement-bookkeeping approaches \citep{LeungBarfootLiu2009}—are consistent and can reproduce centralized performance but incur per-agent communication and memory requirements that scale with the number of agents, making them infeasible on an ultra large scale. Second, well-performing consistent CL can be achieved by leveraging decentralized filtering solutions for linear time-varying systems \citep{PedrosoBatistaEtAl2022DistributedKalman}. However, the design of distributed filter gains is centralized and its computational burden still scales with the number of agents. Third, another common strategy is to approximate covariance propagation to reduce complexity \citep{MadhavanFregeneEtAl2002,MadhavanFregeneEtAl2004}. Yet, this approach is prone to double-counting—an underestimation of correlations between measurements that leads to overconfident and inconsistent estimates, as exemplified in \citet{HowardMataricEtAl2003,PanzieriPascucciSetola2006}. In \citet{Martinelli2007}, this idea is extended to a hierarchical extended Kalman filter that partitions the network of $N$ agents into $K$ groups (also called clusters) of $M$ agents. This scheme requires a leader agent in each cluster that fuses all of the measurements of the cluster at a computational cost of $\mathcal{O}(M^3)$ and a central node that estimates all of the leaders' states at a computational cost of $\mathcal{O}(K^3)$. Since $K=N/M$, the computational burden of the central node grows with $\mathcal{O}((N/M)^3)$, which still grows with the number of agents $N$, and thus is not scalable in an ULSS.
Fourth, covariance intersection (CI) methods  \citep{Julier1997,Carrillo-ArceNerurkarGordilloEtAl2013} have been proposed to eliminate double-counting. For instance, \citet{Carrillo-ArceNerurkarGordilloEtAl2013} achieves consistent CL with communication and computational complexities independent of the network size, albeit at the expense of overly conservative estimates. Moreover, this result relies on the assumption that an agent can predict the state of a neighbor using only its own state estimate and a relative measurement, an assumption that does not always hold. Fifth, split covariance intersection (SCI) methods \citep{LiNashashibiEtAl2013,WanasingheEtAl2014,JulierUhlmann2017} were proposed to attenuate the over-conservatism of CI methods by splitting the covariance contributions in correlated and uncorrelated components and intersecting the correlated components only. As a result, SCI successfully reduces conservatism w.r.t.\ CI while maintaining consistency and, thus, it is the preferred option in many applications \citep{LiYang2019,FangLiEtAl2021,FangLiEtAl2025}. Sixth, recent research  focuses on achieving better performance, albeit at the expense of consistency. For instance, \citet{LuftEtAl2016,LuftEtAl2018} proposed a promising CL method that distributively approximates correlations between each pair of agents, while supporting asynchronous communication and measurements. This approach achieves good performance with per-agent communication requirements that do not scale with the number of agents. However, its per-agent memory requirements grow linearly with the number of agents, and although an approximation error analysis is provided, the method is not provably consistent.

\vspace{0.3cm}
\begin{mdframed}[style=callout]
	In summary, there have been some contributions towards CL strategies based on a covariance intersection approach that are suitable for ULSS. Since these are conservative, recently proposed methods focus on achieving better performance. However, this is often obtained at the expense of requiring local resource usage that scales with the number of agents in the network or at the expense of consistency. As a result, there is a gap to develop CL algorithms that are \emph{well-performing}, \emph{consistent}, and \emph{scalable} to an ultra large scale.
\end{mdframed}

\subsection{Contribution}
This paper addresses the gap identified in the state-of-the-art. Specifically, our contribution is twofold:
\begin{itemize}\itemsep0.4em
	\item We state the problem of designing CL filter gains for each agent subject to the ultra large scale feasibility requirements. We show that this design problem can be expressed in an \emph{overlapping covariance intersection} framework, which is a CI-like framework with partial structural information that was recently introduced and analyzed in \citet{PedrosoBatistaEtAl2026OCI}.
	\item We propose a CL algorithm that relies on an \emph{optimal conservative distributed computation of the covariance propagation} given locally known information. We show that the resulting CL algorithm is provably consistent and is scalable to an ultra large scale. Numerical simulations show that its performance is significantly better than a SCI filter.
\end{itemize}

\subsection{Notation}
The sets of $n\times n$ real symmetric positive semidefinite and positive definite matrices are denoted by $S^n_+$ and  $S^n_{++}$, respectively. Moreover, $\mathbf{P} \succ \zeros$ ($\mathbf{P}\succeq \zeros$) denotes that matrix $\mathbf{P} \in \R^{n\times n}$ is positive definite (semidefinite) and $\mathbf{P} \succ \mathbf{Q}$ ($\mathbf{P} \succeq \mathbf{Q}$) denotes that matrix $\mathbf{P}-\mathbf{Q} \in \R^{n\times n}$  is positive definite (semidefinite). Given a matrix $\mathbf{A}\in \R^{n\times m}$, $\mathbf{A}^+$ denotes the Moore-Penrose inverse of $\mathbf{A}$ \citep[Chap.~1.6]{Zhang2005}.

\section{Problem Statement}\label{sec:probelm_statement}

\subsection{Model of a Network of Output-Coupled Systems} \vspace{-0.2cm}

In this paper, we follow the overarching formulation of a multi-agent system as described in \cite[Section~3]{PedrosoBatistaEtAl2025ULSS}. Consider a network of $N$ interconnected systems, $\Scal_i$, each associated with one computing unit $\Tcal_i$, with $i = 1,2,\dots,N$. Each system has decoupled linear dynamics and has access to linear sensor outputs, which are coupled with a set of other subsystems. The topology of this network is, thus, defined by the output couplings between systems, and it may be time-varying, where time is denoted by $k \in \Nz$. Such output coupling topology may be represented by a directed graph $\Gcal_o(k):=(\Vcal,\Ecal_o(k))$, composed of a set $\Vcal = \{1,2,\ldots,N\}$ of vertices and a set $\Ecal_o(k)$ of directed edges. An edge $e$ incident on vertices $i$ and $j$, directed from $j$ towards $i$, is denoted by $e = (j,i)$. For a vertex $i$, its in-neighborhood $\Ical_i^o(k)$ is the set of vertices from which there is an edge in $\Ecal_o(k)$ directed towards $i$ and its out-neighborhood $\Ocal_i^o(k)$ is the set of vertices towards which there is an edge in $\Ecal^o(k)$ directed away from $i$. In this framework, each system is represented by a vertex, i.e., system $\Scal_i$ is represented by node $i$, and if $\Scal_i$ has access to an output that depends on the state of system $\Scal_j$, then this coupling is represented by an edge directed from vertex $j$ towards vertex $i$, i.e., edge $e = (j,i) \in \Ecal_o(k)$. It is important to stress that the direction of the edge matters. Note, for instance, that the fact that $\Scal_i$ has access to an output that depends on the state of system $\Scal_j$ does not, necessarily, imply the converse. Also, given that the goal of each system $\Scal_i$ is to estimate its own state, it is assumed henceforth, without loss of generality, that the output of system $\Scal_i$ depends on its state, i.e., $i\in \Ical_i^o(k)$. No further assumptions are made on the topology of the directed graph.  

We model each system $\mathcal{S}_i$ by the discrete-time linear system%
\begin{equation}\label{eq:sys_model}
\begin{cases}
	\mathbf{x}_i(k+1) = \mathbf{A}_i(k)\mathbf{x}_i(k) + \mathbf{B}_i(k)\mathbf{u}_i(k) + \mathbf{w}_i(k)\\
	\mathbf{y}_i(k) = \sum_{j\in \Ical_i^o(k)}\mathbf{C}_{ij}(k)\mathbf{x}_j(k) + \mathbf{v}_i(k),
\end{cases}
\end{equation}
where $\mathbf{x}_i(k)\in\R^{n_i}$ is the state vector, $\mathbf{u}_i(k)\in \R^{m_i}$ is the input vector, which is assumed to be known, and $\mathbf{y}_i(k)\in\R^{o_i}$ is the output vector, all of system $\Scal_i$; matrices $\mathbf{A}_i(k) \in \R^{n_i\times n_i}$, $\mathbf{B}_i(k)\in \R^{n_i\times m_i}$, and $\mathbf{C}_{ij}(k) \in \R^{o_i \times n_j}$ with $j\in \Ical_i^o(k)$ are time-varying matrices that model the dynamics and output of system $\Scal_i$; vector $\mathbf{w}_i(k)\in\mathbb{R}^{n_i}$ is the process noise, modeled as a zero-mean white Gaussian process with associated covariance matrix $\mathbf{Q}_i(k) \succeq \zeros_{n_i\times n_i}$, vector $\mathbf{v}_i(k) \in\mathbb{R}^{o_i}$ is the observation noise, modeled as a zero-mean white Gaussian process. It is considered that the observation noise vectors {$\mathbf{v}_i(k)$} and $\mathbf{v}_j(k)$ of systems $\Scal_i$ and $\Scal_j$, respectively, are (possibly) correlated if their outputs are coupled with each other, i.e., $j\in\Ical_i^o(k)$ and $i \in \Ical_j^o(k)$. 
To illustrate the need to introduce such a correlation, consider, for example, two satellites with GNSS receivers. It is possible to obtain a very precise common relative position measurement between the two satellites using carrier-phase differential GNSS \citep{MontenbruckEbinumaEtAl2002}, but the fact that the measurement is common couples the observation noise of both satellites.
The observation noise process, thus, follows $\EV[\mathbf{v}_i(k)\mathbf{v}_i^\top(k)] = \mathbf{R}_{ii}(k) \succ \zeros_{o_i\times o_i}$ and
\begin{equation*}
	\EV[\mathbf{v}_i(k)\mathbf{v}_j^\top(k)] = \begin{cases}
		\mathbf{R}_{ij}(k), & j\in\Ical_i^o(k) \;\text{and} \; i \in \Ical_j^o(k)\\
		\zeros_{o_i\times o_j}, &    \text{otherwise},
	\end{cases}
\end{equation*}
where $\mathbf{R}_{ij}(k) \in \R^{o_i\times o_j}$ are time-varying covariance matrices. Note that if the outputs of $\Scal_i$ and $\Scal_j$ are not coupled, i.e., $j\notin \Ical_i^o(k)$ or $i\notin \Ical_j^o(k)$, the observation noise vectors are not coupled. Crucially, interactions between agents in typical ultra large-scale multi-agent systems are local, thus we assume that $|\Ical_i^o(k)|$ is bounded in time and the bound does not grow with the number of systems, i.e., there exists $M\in \N$ such that $|\Ical_i^o(k)| \leq M$ for all $k\in \Nz$ and all $i\in \{1,\ldots,N\}$ and $M$ grows with $\Ocal(1)$ w.r.t.\ $N$.

The computational unit of each system may be able to establish a communication link with the computational unit of other systems to exchange information. We define another (possibly time-varying) directed graph $\Gcal_c(k):=(\Vcal,\Ecal_c(k))$ that models the topology of the communication network. For each vertex $i\in \Vcal$, one can also define the in-neighborhood $\Ical^c_i(k)$ and  out-neighborhood $\Ocal^c_i(k)$. In line with \citet{PedrosoBatistaEtAl2025ULSS}, we follow the convention that $\Gcal_c(k)$ does not have self-loops for $k\in \Nz$. Crucially, since the possibility of establishing a communication link is especially likely between systems that are output-coupled, in this paper we focus on the particular case whereby the communication topology follows the output coupling topology. Specifically, the communication graph $\Gcal_c(k)$ is composed of all the (nonself) edges of $\Gcal_o(k)$ for $k\in \Nz$.

\subsection{Distributed Cooperative Localization} \vspace{-0.2cm}

In this paper, CL is achieved by implementing a local dynamical filter in each system, based on prediction-update steps employed in a Kalman filter. In a centralized configuration, each system would have access to all systems' output and state estimates, at the expense of all-to-all communication via a central system. However, it is well-known that transmitting outputs and state estimates among systems without significant delay is rarely feasible, so communication is restricted to $\Gcal_c(k)$ and we consider a  distributed configuration instead. Specifically, at each discrete-time instant $k$, the only output known to $\Scal_i$ is $\mathbf{y}_i(k)$. Specifically, we focus on local Luenberger observers, whose prediction and update steps are given by
\begin{equation}\label{eq:dist_filter}
	\!\!\!\begin{cases}
		\hat{\mathbf{x}}_i(k|k\!-\!1) =  \mathbf{A}_i(k\!-\!1)\hat{\mathbf{x}}_i(k\!-\!1) + \mathbf{B}_i(k\!-\!1)\mathbf{u}_i(k\!-\!1)\!\!\!\!\!\\
		\hat{\mathbf{x}}_i(k|k) = 	\mathbf{\hat{x}}_i(k|k-1) + \mathbf{K}_i(k)\!\left(\mathbf{y}_i(k) -  \sum_{j\in \Ical^o_i(k)} \!\mathbf{C}_{ij}(k)\hat{\mathbf{x}}_j(k|k\!-\!1)\right)\!,
	\end{cases}\!\!\!\!\!\!\!\!\!
\end{equation}
where $\hat{\mathbf{x}}_i(k|k-1)$ denotes the predicted state estimate, $\hat{\mathbf{x}}_i(k|k)$ denotes the updated state estimate, and $\mathbf{K}_i(k)$ denotes the local filter gain for system $\Scal_i$ at time $k$. Fig.~\ref{fig:sys_blk_draft} shows a block diagram of the CL solution from the perspective of system $\Scal_i$. From \eqref{eq:dist_filter} and Fig.~\ref{fig:sys_blk_draft}, one can already note the necessity of the communication links in $\Gcal_c(k)$ for the transmission of $\hat{\mathbf{x}}_j(k|k-1)$ at time $k-1$ from $\Tcal_j$ with $j\in \Ical_i^o(k)$ to $\Tcal_i$.

\begin{figure}[ht]
	\centering
	\begin{tikzpicture}[auto, node distance=1.2cm,>=latex']
		\node [block] (Si) {$\mathcal{S}_i$};
		\node [block, below=of Si, yshift=+.5cm] (Ti) {$\mathcal{T}_i$};
		\node [left=of Si.155] (wi) {$\mathbf{w}_i(k-1), \mathbf{v}_i(k)$}; 
		\node [left=of Si.205] (ui) {};
		\node [above=of Si,yshift=-.5cm, label={[align=center, yshift=-0.3cm] above:$\mathbf{x}_j(k),\: j\in \Ical_i^o(k)$}] (xj) {}; 
		\node [left=of Ti.205] (Tj-in) {$\mathcal{T}_j,\: j\in \mathcal{I}_i^c(k)\;$};
		\node [below=of Ti, yshift=+.5cm] (Tj-out) {$\;\mathcal{T}_j, \:j\in \mathcal{O}_i^c(k)$};
		\node [right=of Ti.335] (xi-hat) {$\;\hat{\mathbf{x}}_i(k)$};
		\draw[->] (wi) -- (Si.155);
		\draw[->] (Tj-in) -- (Ti.205);
		\draw[->] (Ti) -- (Tj-out) ;
		\draw[->] (Ti.335) -- (xi-hat) ;
		\draw[->] (xj) -- (Si) ;
		\draw[dashed,->] (Ti.155) -- ++ (-1,0) |- node [pos=0.25] {$\mathbf{u}_i(k-1)$}  (Si.205) ;
		\draw[->] (Si.335) -- ++ (1,0) |- node [pos=0.25] {$\mathbf{y}_i(k)$}  (Ti.25) ;
	\end{tikzpicture}    
	\caption{Block diagram of the CL solution from the perspective of $\mathcal{S}_i$.  Although this is outside the scope of the paper, the dashed line depicts that the input to system $\Scal_i$ could be computed using the estimate $\hat{\mathbf{x}}_i(k)$.}
	\label{fig:sys_blk_draft}
\end{figure}
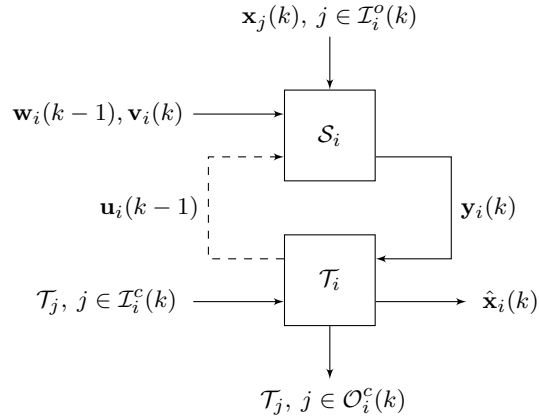

\begin{remark}
	It is well-known that, even for linear time-invariant (LTI) systems, linear controllers may not be optimal among the family of distributed controllers \citep{Witsenhausen1968}. As a result, the filter \eqref{eq:dist_filter} is suboptimal in general. Indeed, it is shown in \citet{PedrosoBatista2023DecentralizedMHE} that, even for LTI systems, a linear moving horizon estimation approach may unlock significant performance. Nevertheless, given their simplicity and widespread use, we focus on filters of the form \eqref{eq:dist_filter}. \hfill$\triangle$
\end{remark}

We now turn to the concept of consistency of a filter. Define the estimation error $\tilde{\mathbf{x}}_i(k|k-1) := \mathbf{\hat{x}}_i(k|k-1)-\mathbf{x}_i(k)$ and $\tilde{\mathbf{x}}_i(k|k) := \hat{\mathbf{x}}_i(k|k)-\mathbf{x}_i(k)$. Intuitively, a local filter is consistent if the estimation error covariance matrices provided by the filter are conservative w.r.t.\ the true estimation error covariance. A formal definition is presented below.

\begin{definition}\label{def:consistency}
	A filter is said to be \emph{consistent} (also sometimes called \emph{conservative}) if, at each time $k$, it computes an unbiased estimate $\hat{\mathbf{x}}_i(k|k)$ of $\mathbf{x}_i(k)$ and an estimation error covariance matrix $\mathbf{X}_{ii}(k|k) \succeq \zeros$ that is a bound on the true estimation covariance, i.e., $\mathbf{X}_{ii}(k|k)\succeq \EV[\tilde{\mathbf{x}}_i(k|k)\tilde{\mathbf{x}}_i^\top(k|k)]$. \hfill$\triangle$
\end{definition}

\begin{remark}
	The synchronization of the local filters \eqref{eq:dist_filter} can be readily implemented when all systems have access to a common external signal. Such signals are common in practical applications, and their use does not violate the ULSS requirements. For example, if each system is equipped with a GNSS receiver, GNSS signals can be used for precise synchronization. In challenging environments where such signals are typically unavailable, such as underwater applications, synchronization requires complex protocols such as clock disciplining. \hfill$\triangle$
\end{remark}



Naturally, the goal is to design, for each system $\Scal_i$ with $i = 1,2,\ldots,N$ and at each time $k\in \N$, the local filter gains $\mathbf{K}_i(k)$ such that: (i)~the filter is consistent; and (ii)~the covariance bounds $\mathbf{X}_{ii}(k|k)$ are optimized according to a prescribed performance criterion. We formulate this design problem next.

\subsection{Distributed Design Problem} \vspace{-0.2cm}

This paper focuses on networks of systems when $N$ is large enough such that a centralized design procedure is not computationally feasible, i.e., it is not possible to centrally synthesize the gains of the distributed filter in \eqref{eq:dist_filter}. These systems are called \emph{ultra large-scale systems}, a definition that was recently introduced in \citet{PedrosoBatistaEtAl2025ULSS}. The synthesis stage of this emerging class of systems is notoriously challenging and is subject to stringent requirements for a scalable deployment. Briefly, in this case, the synthesis stage entails \emph{locally designing} the filter gains $\mathbf{K}_i(k)$ in real time and with limited communication, computational, and memory resources according to a performance criterion. The distributed design problem can then be formally stated from the perspective of each computational unit $\Tcal_i$ as follows.

\begin{problem}\label{prob:design}
	At each time $k$, in each computational unit $\Tcal_i$, 
	with knowledge of: \vspace{-0.2cm}
	\begin{itemize}\itemsep0.2em 
		\item The systems with which $\Scal_i$ has output couplings at time $k$, i.e., $\Ical^o_i(k)$;
		\item The input at time $k-1$, i.e., $\mathbf{u}_i(k-1)$;
		\item The dynamic matrices $\mathbf{A}_i(k\!-\!1)$,  $\mathbf{B}_i(k\!-\!1), \mathbf{Q}_i(k\!-\!1), \mathbf{C}_{ij}(k)$ with $j\in \Ical^o_i(k)$, and $\mathbf{R}_{ij}(k)$ with $j\in \Ical^o_i(k)$;
		\item Information received from $\Tcal_j$ with $j\in \Ical^o_i(k)$ ; \vspace{-0.2cm}
	\end{itemize}
	design a gain $\mathbf{K}_i(k)$ to compute:\vspace{-0.2cm}
	\begin{itemize}\itemsep0.2em
		\item An unbiased estimate $\hat{\mathbf{x}}_i(k|k)$ of $\mathbf{x}_i(k)$, resorting to the distributed filter \eqref{eq:dist_filter};
		\item A covariance bound $\mathbf{X}_{ii}(k|k)\succeq \EV[\tilde{\mathbf{x}}_i(k|k) \tilde{\mathbf{x}}_i^\top(k|k)]$;
	\end{itemize}
	such that $J_i(\mathbf{X}_{ii}(k|k))$ is minimized, where $J_i:S_+^{n_i} \to \R$ is a performance metric. Furthermore, the design procedure must follow the ultra large-scale feasibility requirements in \citet[Table~1]{PedrosoBatistaEtAl2025ULSS}, i.e., the computational, memory, and communication load in $\Tcal_i$ must not grow with $N$ and $\Tcal_i$ cannot use information from $\Tcal_j$ with $j\in \Ical^o_i(k)$ that has been computed in $\Tcal_j$ at time $k$ and transmitted through $\Gcal_c$.
	
\end{problem}

The design problem in Problem~\ref{prob:design} is \emph{distributed} in the sense that the synthesis of the filter gains is distributed across the computational units of systems in the network. Specifically, resorting to local information, each computational unit designs a local gain $\mathbf{K}_i(k)$ to obtain an unbiased and consistent estimate of the local state.
The performance metrics $J_i$ with $i = 1,2,\ldots,N$ are assumed to satisfy the following monotonicity criterion.

\begin{assumption}\label{ass:J}
	For all $i = 1,2,\ldots,N$, given $\mathbf{X}, \mathbf{Y} \in S^{n_i}_+$, the map $J_i:S_+^{n_i} \to \R$ is such that $\mathbf{X} \succ  \mathbf{Y} \implies J_i(\mathbf{X}) > J_i(\mathbf{Y})$. \hfill$\triangle$
\end{assumption}

This monotonicity assumption on $J_i$ is very mild. Intuitively, let $\mathbf{B}_1$ and $\mathbf{B}_2$ be error covariance matrices. Assumption~\ref{ass:J} enforces that if $\mathbf{B}_1 \prec \mathbf{B}_2$, i.e., the covariance $\mathbf{B}_2$ portrays a larger spread of the error distribution in every direction than $\mathbf{B}_1$, then $J_i(\mathbf{B}_2)>J_i(\mathbf{B}_1)$. Common criteria in fusion applications such as the trace or determinant satisfy Assumption~\ref{ass:J}.

\begin{remark}
	Note that if the dynamics of \eqref{eq:sys_model} are nonlinear, an extended Kalman filter framework also leads to the design problem in Problem~\ref{prob:design}. In that case, the dynamic matrices are instead obtained by linearizing the nonlinear dynamics about the last state estimate. It is worth remarking that the analysis of the consistency of the filter becomes significantly more complex. That is because the estimation error distribution is not fully characterized by its second moment anymore, even if the process and sensor noise are Gaussian distributed. In any case, the same consideration applies in the centralized setting. \hfill$\triangle$
\end{remark}


\section{Distributed Synthesis as an Overlapping Covariance Intersection Problem}\label{sec:distributed_synthesis}

In Section~\ref{sec:MVE}, we show that the design of local minimum variance estimation gains, which is a particular case of Problem~\ref{prob:design} when the performance criterion is $J_i(\cdot) = \tr(\cdot)$, has local information requirements that violate the ultra large scale requirements of Problem~\ref{prob:design}. In Section~\ref{sec:linear_estimator}, we formulate the filter design problem as a problem of designing an optimal linear estimator. In Section~\ref{sec:partial_info}, we analyze the partial information that is available in each computational unit for the design problem. In Section~\ref{sec:OCI}, we cast the distributed design problem as a generalized covariance intersection problem, which expresses the fusion gains as the minimization of a performance metric $J_i(\cdot)$ of the worst-case uncertainty under the information that is locally available.


\subsection{Local Information is Insufficient for Minimum Variance Estimation}\label{sec:MVE} \vspace{-0.2cm}

We consider that at each time $k$, each $\Tcal_i$ has access to (i)~an unbiased estimate $\hat{\mathbf{x}}_i(k|k-1)$ of $\mathbf{x}_i(k)$ (obtained from the application of the CL algorithm in the previous instant); (ii)~unbiased estimates $\hat{\mathbf{x}}_j(k|k-1)$ of $\mathbf{x}_j(k)$ with $j \in \Ical^o_i(k)$ obtained through communication; and (iii)~the output $\mathbf{y}_i(k)$. It is well-known \citep[Theorem~3.1 and Example~3.3]{AndersonMoore1979} that at each time $k$, the local minimum-variance estimate, i.e., with $J_i(\cdot) = \tr(\cdot)$, is given by \eqref{eq:dist_filter} with the optimal gain
\begin{equation}\label{eq:optimal_K}
	\mathbf{K}_i^{\mathrm{MV}}(k) = \EV[\tilde{\mathbf{x}}_i(k|k-1)\tilde{\mathbf{\bepsilon}}_i^\top(k)]\EV[\tilde{\mathbf{\bepsilon}}_i(k) \tilde{\mathbf{\bepsilon}}_i^\top(k)]^{-1},
\end{equation}
where $\tilde{\bepsilon}_i(k) := \mathbf{y}_i(k) -  \sum_{j\in \Ical^o_i(k)} \mathbf{C}_{ij}(k)\hat{\mathbf{x}}_j(k|k-1)$. However, notice that knowledge about the covariance matrices $\EV[\tilde{\mathbf{x}}_i(k|k-1)\tilde{\mathbf{\bepsilon}}_i^\top(k)]$ and $\EV[\tilde{\mathbf{\bepsilon}}_i(k) \tilde{\mathbf{\bepsilon}}_i^\top(k)]$ is required for the synthesis of the optimal minimum-variance local gain $\mathbf{K}_i^{\mathrm{MV}}(k)$ according to \eqref{eq:optimal_K}. However, without all-to-all communication (which violates the ultra large-scale synthesis requirements), it is not possible to compute these covariance matrices in $\Tcal_i$ at time $k$. To see why, define $\mathbf{P}_{ij}(k|k) := \EV[\tilde{\mathbf{x}}_i(k|k) \tilde{\mathbf{x}}_j^\top(k|k)]$ and $\mathbf{P}_{ij}(k|k-1) := \EV[\tilde{\mathbf{x}}_i(k|k-1)\tilde{\mathbf{x}}_j^\top(k|k-1)]$. Assuming that the initial estimates $\tilde{\mathbf{x}}_i(0)$, for $i \in \{1,2,\ldots,N\}$, are unbiased and have covariance matrices $\mathbf{P}_{ij}(0|0)$, for $i,j\in \{1,2,\ldots,N\}$, from the system model in \eqref{eq:sys_model} and the filter in \eqref{eq:dist_filter} one can obtain the well-known recurrence characterized by
\begin{equation}\label{eq:propagation_P_pred}
	\begin{split}
		\mathbf{P}_{pq}(k|k\!-\!1) = \delta_{pq}\mathbf{Q}_{pp}(k\!-\!1) + \mathbf{A}_{pp}(k\!-\!1)\mathbf{P}_{pq}(k\!-\!1|k\!-\!1)\mathbf{A}_{qq}^\top(k\!-\!1) 
	\end{split}
\end{equation}
and
\begin{equation}\label{eq:propagation_P}
	\begin{split}
		\mathbf{P}_{pq}(k|k) = \mathbf{K}_p(k)\mathbf{R}_{pq}(k)\mathbf{K}_q^\top(k)  + \sum\nolimits_{r\in \Ical_p^o(k)}\sum\nolimits_{s\in \Ical_q^o(k)}(\delta_{pr}\mathbf{I}-\mathbf{K}_p(k)\mathbf{C}_{pr}(k))\mathbf{P}_{rs}(k|k-1)(\delta_{qs}\mathbf{I}-\mathbf{K}_q(k)\mathbf{C}_{qs}(k))^\top,
	\end{split}
\end{equation}
which can be employed to write the expectations in \eqref{eq:optimal_K} as
\begin{equation*}
	\begin{split}
		\EV[\tilde{\mathbf{\bepsilon}}_i(k)\tilde{\mathbf{\bepsilon}}_i^\top(k)] = \mathbf{R}_{ii}(k)+ \sum\nolimits_{p\in \Ical_i^o(k)}\sum\nolimits_{q\in \Ical_i^o(k)}\mathbf{C}_{ip}(k)\mathbf{P}_{pq}(k|k-1)\mathbf{C}_{iq}^\top(k)
	\end{split}
\end{equation*}
and
\begin{equation*}
	\EV[\tilde{\mathbf{x}}_i(k|k-1)\tilde{\mathbf{\bepsilon}}_i^\top(k)] = \sum\nolimits_{p\in \Ical_i^o(k)}\mathbf{P}_{ip}(k|k-1)\mathbf{C}_{ip}^\top(k).
\end{equation*}
To compute the minimum-variance gain \eqref{eq:optimal_K} locally in $\Tcal_i$, one would have to compute $\mathbf{P}_{pq}(k|k-1)$ and, from \eqref{eq:propagation_P_pred}, have access to $\mathbf{P}_{pq}(k-1|k-1)$ with $p,q\in \Ical_i^o(k)$ either by computing them in $\Tcal_i$ or by communication with other computational units. In turn, computing $\mathbf{P}_{pq}(k-1|k-1)$ with $p,q\in \Ical_i^o(k)$ would require access (by computing it or by communication) to $\mathbf{P}_{rs}(k-1|k-2)$ with $(r,s)\in \Ical_p^o(k) \times  \Ical_q^o(k)$, and so forth. Because of that, in general, the exact cooperative propagation of $\mathbf{P}_{pq}(k)$ is not possible without all-to-all communication or memory and computational requirements that grow with $N$, which violates the synthesis requirements in Problem~\ref{prob:design}.

In this paper, we follow an approach inspired by \citet{PedrosoBatista2023DistributedEKF} and novel CI tools \citep{PedrosoBatistaEtAl2026OCI,PedrosoBatistaEtAl2026UnificationCI} to devise a distributed synthesis procedure that is consistent. 

\vspace{0.3cm}
\begin{mdframed}[style=callout]
	The main idea is that each computational unit $\Tcal_i$ keeps and updates a bound on error covariances between the systems in $\Ical_i^o$. This information alone is not enough to propagate $\mathbf{P}_{pq}(k-1|k-1)$ with $p,q\in \Ical_i^o$, exactly, as aforementioned. Instead, we use the concept of \emph{robustness to information uncertainty}. Specifically, the local gains are designed to minimize the worst-case consistent bound on $\EV[\tilde{\mathbf{x}}_i(k|k)\tilde{\mathbf{x}}_i^\top(k|k)]$ given the information that is known. This problem is cast as a novel formulation of a covariance intersection problem~\citep{ForslingNoackEtAl2024}, so-called overlapping covariance intersection (OCI)~\citep{PedrosoBatistaEtAl2026OCI}.
\end{mdframed}

\subsection{Formulation as Optimal Linear Estimator}\label{sec:linear_estimator} \vspace{-0.2cm}

To cast the filtering problem as an OCI problem \citep{PedrosoBatistaEtAl2026OCI} we follow a framework that is closer to the signal processing literature. Specifically, consider at time $k$ a computational unit $\Tcal_i$ whose goal is to obtain an estimate $\hat{\mathbf{x}}_i(k|k)$. It has access to two unbiased (partial) estimates of $\hat{\mathbf{x}}_i(k|k)$. The first estimate is $\mathbf{z}_{i,1}(k) := \hat{\mathbf{x}}_i(k|k-1) = \mathbf{A}_i(k-1)\hat{\mathbf{x}}_i(k-1|k-1) + \mathbf{B}_i(k-1)\mathbf{u}_i(k-1)$, where $\hat{\mathbf{x}}_i(k-1|k-1)$ is obtained from the application of the CL algorithm in the previous instant. Notice that the estimate $\mathbf{z}_{i,1}(k)$ is unbiased because 
\begin{equation}\label{eq:z1_unbiased}
	\begin{split}
		\!\!\mathbf{e}_{i,1}(k) &= \mathbf{z}_{i,1}(k)-\mathbf{x}_i(k) \\
		& = \mathbf{A}_i(k\!-\!1)\hat{\mathbf{x}}_i(k\!-\!1|k\!-\!1) \!+\! \mathbf{B}_i(k\!-\!1)\mathbf{u}_i(k\!-\!1)- 	\mathbf{A}_i(k\!-\!1)\mathbf{x}_i(k\!-\!1) - \mathbf{B}_i(k\!-\!1)\mathbf{u}_i(k\!-\!1) - \mathbf{w}_i(k\!-\!1)\\
		& = \mathbf{A}_i(k-1)\tilde{\mathbf{x}}_i(k-1|k-1) - \mathbf{w}_i(k-1) \\
		& = \tilde{\mathbf{x}}_i(k|k-1) 
	\end{split}
\end{equation}
and $\tilde{\mathbf{x}}_i(k-1|k-1)$ and $\mathbf{w}_i(k-1)$ are unbiased by hypothesis. The second estimate is
\begin{equation*}
	\mathbf{z}_{i,2}(k):= \mathbf{y}_i(k)-\sum\nolimits_{p\in \Ical^o_i(k)\setminus \{i\}} \mathbf{C}_{ip}(k)\hat{\mathbf{x}}_p(k|k-1),
\end{equation*}
which is unbiased because 
\begin{equation}\label{eq:z2_unbiased}
	\begin{split}
		\!\!\mathbf{e}_{i,2}(k) &=  \mathbf{z}_{i,2}(k)- \mathbf{C}_{ii}(k)\mathbf{x}_i(k) \\
		& = \mathbf{y}_i(k)- \sum\nolimits_{p\in \Ical^o_i(k)\setminus \{i\}}  \mathbf{C}_{ip}(k)\hat{\mathbf{x}}_p(k|k-1) - \mathbf{C}_{ii}(k)\mathbf{x}_i(k)\!\!\!\!\!\! \\
		& = \sum\nolimits_{p\in \Ical^o_i(k)}  \mathbf{C}_{ip}(k)\mathbf{x}_p(k) + \mathbf{v}_i(k) 
		-\sum\nolimits_{p\in \Ical^o_i(k)\setminus \{i\}}\mathbf{C}_{ip}(k)\hat{\mathbf{x}}_p(k|k\!-\!1) \!-\! \mathbf{C}_{ii}(k)\mathbf{x}_i(k)\!\!\\
		& = -\sum\nolimits_{p\in \Ical^o_i(k)\setminus \{i\}} \mathbf{C}_{ip}(k)\tilde{\mathbf{x}}_p(k|k-1)+ \mathbf{v}_i(k)
	\end{split}
\end{equation}
and  $\tilde{\mathbf{x}}_p(k|k-1)$ and $\mathbf{v}_i(k)$ are unbiased by hypothesis. Concatenating the two (partial) estimates one can write
\begin{equation*}
	 \underbrace{\begin{bmatrix}
		\mathbf{z}_{i,1}(k) \\
		\mathbf{z}_{i,2}(k)
	\end{bmatrix}}_{=: \mathbf{z}_{\Ical_i}(k)} = \underbrace{\begin{bmatrix}
		\mathbf{I}_{n_i} \\
		\mathbf{C}_{ii}(k)
	\end{bmatrix}}_{=:\mathbf{H}_{\Ical_i}(k)} \mathbf{x}_i(k) + \underbrace{\begin{bmatrix}
	\mathbf{e}_{i,1}(k) \\
	\mathbf{e}_{i,2}(k)
\end{bmatrix}}_{=:\mathbf{e}_{\Ical_i}(k)} ,
\end{equation*}
where $\EV[\mathbf{e}_{\Ical_i}(k)] = \zeros$ and, from \eqref{eq:z1_unbiased} and \eqref{eq:z2_unbiased}, $\mathbf{e}_{\Ical_i}(k)$ can be written as
\begin{equation}\label{eq:e_I_def}
	\mathbf{e}_{\Ical_i}(k) \!:=\! \begin{bmatrix}
		\tilde{\mathbf{x}}_i(k|k-1)\\ -{\sum_{p\in \Ical^o_i(k)\setminus \{i\}}}  \mathbf{C}_{ip}(k)\tilde{\mathbf{x}}_p(k|k-1)
	\end{bmatrix} + \begin{bmatrix}
		\zeros\\
		\mathbf{v}_i(k)
	\end{bmatrix}\!.
\end{equation}

First, a linear estimator for $\mathbf{x}_i(k)$ is of the form $\hat{\mathbf{x}}_i(k|k) = \mathbf{K}_{\Ical_i}(k)\mathbf{z}_{\Ical_i}(k)$ and it is unbiased if and only if $\mathbf{K}_{\Ical_i}(k)  \mathbf{H}_{\Ical_i}(k) = \mathbf{I}$. One can now expand the condition on the gain as
\begin{equation}\label{eq:gain_OCI_blocks}
	\begin{split}
		\underbrace{\begin{bmatrix}
				\tilde{\mathbf{K}}_i(k) & 
				\mathbf{K}_i(k)
		\end{bmatrix}}_{\mathbf{K}_{\Ical_i}(k)} \begin{bmatrix}
			\mathbf{I}_{n_i} \\
			\mathbf{C}_{ii}(k)
		\end{bmatrix} = \mathbf{I} \iff \tilde{\mathbf{K}}_i(k) = \mathbf{I}-\mathbf{K}_i(k)\mathbf{C}_{ii}(k).
	\end{split}
\end{equation}
Substituting \eqref{eq:gain_OCI_blocks} in the expression for the linear estimator yields
\begin{equation}\label{eq:unbiased_dist_filter}
	\begin{split}
		\hat{\mathbf{x}}_i(k|k) & = (\mathbf{I}-\mathbf{K}_i(k)\mathbf{C}_{ii}(k))\mathbf{z}_{i,1}(k) + \mathbf{K}_i(k)\mathbf{z}_{i,2}(k)\\
		& =  (\mathbf{I}-\mathbf{K}_i(k)\mathbf{C}_{ii}(k))\hat{\mathbf{x}}_i(k|k-1)+ \mathbf{K}_i(k) \bigg(\mathbf{y}_i(k)-\sum\nolimits_{p\in \Ical^o_i(k)\setminus \{i\}}  \mathbf{C}_{ip}(k)\hat{\mathbf{x}}_p(k|k-1)\bigg)\!\!\\
		& = 	\mathbf{\hat{x}}_i(k|k-1) + \mathbf{K}_i(k)\bigg(\mathbf{y}_i(k) - \sum\nolimits_{j\in \Ical^o_i(k)} \mathbf{C}_{ij}(k)\hat{\mathbf{x}}_j(k|k-1)\bigg),
	\end{split}
\end{equation}
which, unsurprisingly, is exactly the same expression as the Luenberger filter expression in \eqref{eq:dist_filter}. 

The estimation error covariance of the estimate $\hat{\mathbf{x}}_i(k|k)$ is denoted as $\mathbf{P}_{ii}(k|k):= \EV[\tilde{\mathbf{x}}_i(k|k)\tilde{\mathbf{x}}_i^\top (k|k)]$, and it can be expressed after algebraic manipulation using  $\hat{\mathbf{x}}_i(k|k) = \mathbf{K}_{\Ical_i}(k)\mathbf{z}_{\Ical_i}(k)$ as
\begin{equation}\label{eq:Pii}
	\mathbf{P}_{ii}(k|k) = \mathbf{K}_{\Ical_i}(k)  \EV[\mathbf{e}_{\Ical_i}(k)\mathbf{e}_{\Ical_i}^\top(k)]  \mathbf{K}_{\Ical_i}^\top(k).
\end{equation} 
According to Problem~\ref{prob:design}, one desires to design $\mathbf{K}_{\Ical_i}(k)$ such that $J_i(\mathbf{P}_{ii}(k|k))$ is minimized. We have now cast the filter design problem as finding  $\mathbf{K}_{\Ical_i}(k)$ such that $J_i(\mathbf{P}_{ii}(k|k))$ is  minimized subject to $\mathbf{K}_{\Ical_i}(k) \mathbf{H}_{\Ical_i}(k) = \mathbf{I}$. However, from \eqref{eq:Pii}, $\mathbf{P}_{ii}(k|k)$ depends on $\EV[\mathbf{e}_{\Ical_i}(k)\mathbf{e}_{\Ical_i}^\top(k)]$, which is not exactly known without all-to-all communication, as analyzed in Section~\ref{sec:MVE}. In what follows, we characterize the partial information about $\EV[\mathbf{e}_{\Ical_i}(k)\mathbf{e}_{\Ical_i}^\top(k)]$ that is available in $\Tcal_i$.


\subsection{Partial Information about
\texorpdfstring{$\EV[\mathbf{e}_{\Ical_i}(k)\mathbf{e}_{\Ical_i}^\top(k)]$}
{expected error covariance}}\label{sec:partial_info} \vspace{-0.2cm}

Notice that using \eqref{eq:unbiased_dist_filter}, the estimation errors $\tilde{\mathbf{x}}_p(k|k-1)$ with $p\in \Ical_i^o(k)$ can be expanded as
\begin{equation}\label{eq:tilde_x_p}
	\begin{split}
		\tilde{\mathbf{x}}_p(k|k-1) & = \mathbf{A}_p(k-1)\tilde{\mathbf{x}}_p(k-1|k-1) - \mathbf{w}_p(k)\\
		& = \mathbf{A}_p(k\!-\!1)\!\!\!\!\sum_{r\in \Ical^o_p(k-1)}\!\!\!\!(\delta_{pr}\mathbf{I}-\mathbf{K}_p(k\!-\!1) \mathbf{C}_{pr}(k\!-\!1))\tilde{\mathbf{x}}_r(k\!-\!1|k\!-\!2)+ \mathbf{A}_p(k\!-\!1)\mathbf{K}_p(k\!-\!1)\mathbf{v}_p(k\!-\!1) -  \mathbf{w}_p(k).
	\end{split}
\end{equation}
Denote $\tilde{\mathbf{x}}_{\Ical_i}(k|k-1) := \col(\tilde{\mathbf{x}}_p(k|k-1); p\in \Ical_i^o(k))$, $\tilde{\mathbf{x}}_{\Ical_i^2}(k-1|k-2) := \col(\tilde{\mathbf{x}}_r(k-1|k-2); r\in \Ical_p^o(k-1), p\in \Ical_i^o(k))$, and the respective covariance matrices as $\mathbf{P}_{\Ical_i}(k|k-1):= \EV[\tilde{\mathbf{x}}_{\Ical_i}(k|k-1)\tilde{\mathbf{x}}_{\Ical_i}^\top(k|k-1)]$ and $\mathbf{P}_{\Ical_i^2}(k-1|k-2):= \EV[\tilde{\mathbf{x}}_{\Ical^2_i}(k-1|k-2)\tilde{\mathbf{x}}_{\Ical_i^2}^\top(k-1|k-2)]$. Notice that $\mathbf{P}_{\Ical_i}(k|k-1)$ represents the estimation error covariance between neighbors of $\Scal_i$, i.e., $p\in \Ical_i^o(k)$, and $\mathbf{P}_{\Ical_i^2}(k|k-1)$ represents the estimation error covariance between neighbors of neighbors of $\Scal_i$, i.e., $r\in  \Ical_p^o(k-1)$ and $p \in  \Ical_i^o(k)$. By \eqref{eq:tilde_x_p} one can write
\begin{equation}\label{eq:P_I_2}
	\mathbf{P}_{\Ical_i^2}(k|k-1) = \mathbf{D}_{\Ical_i}(k-1) \mathbf{P}_{\Ical_i^2}(k-1|k-2)\mathbf{D}^\top_{\Ical_i}(k-1)+ \mathbf{E}_{\Ical_i}(k-1), 
\end{equation}
where the expressions for $\mathbf{D}_{\Ical_i}(k\!-\!1)$ and $\mathbf{E}_{\Ical_i}(k\!-\!1) \succeq \zeros$ follow immediately from \eqref{eq:tilde_x_p}.\footnote[1]{The expressions for $\mathbf{D}_{\Ical_i}(k-1)$, $\mathbf{E}_{\Ical_i}(k-1)$, $\mathbf{C}_{\Ical_i}(k)$, and $\mathbf{R}_{\Ical_i}(k)$ are not presented for the sake of conciseness and clarity. Nevertheless, these expressions are implemented in an open-access repository along with the code of the numerical simulations presented in Section~\ref{sec:results}.} Likewise, from \eqref{eq:e_I_def} and \eqref{eq:P_I_2}, one can express $\EV[\mathbf{e}_{\Ical_i}(k)\mathbf{e}_{\Ical_i}^\top(k)]$ as
\begin{equation}\label{eq:Eee}
\EV[\mathbf{e}_{\Ical_i}(k)\mathbf{e}_{\Ical_i}^\top(k)] = \mathbf{R}_{\Ical_i}(k) + \mathbf{C}_{\Ical_i}(k)\mathbf{P}_{\Ical_i^2}(k\!-\!1|k\!-\!2)\mathbf{C}^\top_{\Ical_i}(k),
\end{equation}
where the expressions for $\mathbf{C}_{\Ical_i}(k)$ and $\mathbf{R}_{\Ical_i}(k) \succeq \zeros$ follow immediately from \eqref{eq:e_I_def} and \eqref{eq:P_I_2}.\footnotemark~The formulation of the design problem as a generalized covariance intersection problem requires that $\mathbf{R}_{\Ical_i}(k)$ is positive definite. That can be guaranteed under the following mild assumption.

\begin{assumption}\label{ass:Q_pd}
	For all $k\in \Nz$ and for all $i= 1,2,\ldots, N$, $\mathbf{Q}_i(k) \succ \zeros$. \hfill$\triangle$
\end{assumption}

\begin{lemma}\label{lem:R_pd}
	Under Assumption~\ref{ass:Q_pd}, $\mathbf{R}_{\Ical_i}(k) \succ \zeros$ for all $i = 1,2,\ldots,N$ and all $k\in \N$.
\end{lemma}\vspace{-0.8cm}
\begin{pf}
	Since $\tilde{\mathbf{x}}_i(k|k-1) = \mathbf{A}_i(k-1)\tilde{\mathbf{x}}_i(k-1|k-1) - \mathbf{w}_i(k-1)$, one can rewrite \eqref{eq:e_I_def} as
\begin{equation*}
	\mathbf{e}_{\Ical_i}(k) := \begin{bmatrix}
		\mathbf{A}_i(k)\tilde{\mathbf{x}}_i(k-1|k-1)\\ -\sum\nolimits_{p\in \Ical^o_i(k)\setminus \{i\}} \mathbf{C}_{ip}(k)(\mathbf{A}_p(k-1)\tilde{\mathbf{x}}_p(k-1|k-1) - \mathbf{w}_p(k))
	\end{bmatrix} + \begin{bmatrix}
		-\mathbf{w}_i(k-1)\\
		\mathbf{v}_i(k)
	\end{bmatrix}.
\end{equation*}
	Therefore, from the definition of $\mathbf{R}_{\Ical_i}(k)$ in \eqref{eq:Eee}, it follows that $\mathbf{R}_{\Ical_i}(k) \!\succeq \! \EV \!\left[[-\mathbf{w}_i(k\!-\!1)\!^\top \mathbf{v}_i(k)\!^\top ]\!^\top[-\mathbf{w}_i(k\!-\!1)\!^\top \mathbf{v}_i(k)\!^\top ]\right]$, i.e., $\mathbf{R}_{\Ical_i}(k) \succeq \diag(\mathbf{Q}_i(k\!-\!1),\mathbf{R}_{ii}(k))$. By hypothesis $\mathbf{R}_{ii}(k) \succ \zeros$ and, by Assumption~\ref{ass:Q_pd}, $\mathbf{Q}_i(k-1)\succ \zeros$, which immediately proves the result. \hfill$\square$
\end{pf}

\begin{mdframed}[style=callout]
	In this paper, we propose that, at time $k$, each computational unit $\Tcal_i$ maintains a bound on $\mathbf{P}_{\Ical_i}(k-1|k-2)$, which is denoted by $\mathbf{X}_{\Ical_i}(k-1|k-2) \succeq \mathbf{P}_{\Ical_i}(k-1|k-2)$. To guarantee that such a bound is available at the next instant, then at time $k$ a new bound $\mathbf{X}_{\Ical_i}(k|k-1) \succeq \mathbf{P}_{\Ical_i}(k|k-1)$ must also be computed during the design procedure of $\mathbf{K}_{\Ical_i}(k)$.
\end{mdframed}

\vspace{-.2cm}


At time $k$, the computational unit $\Tcal_i$ can gather the bounds $\mathbf{X}_{\Ical_p}(k-1|k-2) \succeq \mathbf{P}_{\Ical_p}(k-1|k-2)$ with $p\in \Ical_i^o(k)$ from its neighbors. Notice that each $\mathbf{P}_{\Ical_p}(k-1|k-2)$  with $p\in \Ical_i^o(k)$ is a principal submatrix of $\mathbf{P}_{\Ical^2_i}(k-1|k-2)$. Therefore, the computational unit $\Tcal_i$ has access to $|\Ical_i^o(k)|$ bounds on principal submatrices of $\mathbf{P}_{\Ical^2_i}(k-1|k-2)$, which allows us to express a set of admissible matrices $\mathbf{P}_{\Ical_i^2}(k-1|k-2)$ as 
\begin{equation}\label{eq:admissible_P}
	\Pcal_{\Ical^2_i}(k-1) := \left\{ \mathbf{P}_{\Ical_i^2} \succ \zeros : \mathbf{W}_{ip}(k-1) \mathbf{P}_{\Ical_i^2} \mathbf{W}_{ip}^\top (k-1)\preceq \mathbf{X}_{\Ical_p}(k-1|k-2)\; \forall p\in \Ical_i^o(k)  \right\},
\end{equation}
where matrices $\mathbf{W}_{ip}(k-1)$ for $ p\in \Ical_i^o(k) $ extract principal submatrices of $ \mathbf{P}_{\Ical_i^2}$ associated with $\mathbf{X}_{\Ical_p}(k-1|k-2)$, i.e.,  $\mathbf{W}_{ip}(k\!-\!1)$ is such that $\mathbf{W}_{ip}(k\!-\!1) \mathbf{P}_{\Ical_i^2}(k\!-\!1|k\!-\!2) \mathbf{W}_{ip}^\top (k\!-\!1) = \mathbf{P}_{\Ical_p}(k\!-\!1|k\!-\!2)$. 

\vspace{0.2cm}
\begin{mdframed}[style=callout]
	One concludes that $\EV[\mathbf{e}_{\Ical_i}(k)\mathbf{e}_{\Ical_i}^\top(k)]$ is given by \eqref{eq:Eee}, where $\mathbf{P}_{\Ical_i^2}(k-1|k-2)$ is not exactly known but is known to satisfy $\mathbf{P}_{\Ical_i^2}(k-1|k-2) \in 	\Pcal_{\Ical^2_i}(k-1)$ with information locally available in $\Tcal_i$.
\end{mdframed}
\vspace{-.1cm}

Given the proposed information structure, one can show that the set of admissible matrices $\mathbf{P}_{\Ical_i^2}(k-1|k-2)$ is bounded.
\vspace{-.3cm}

\begin{lemma}\label{lem:W_full_col_rank}
		For all $i =1, 2, \ldots, N$, if at time $k$ there exist bounds $\mathbf{X}_{\Ical_p}(k-1|k-2) \succeq \mathbf{P}_{\Ical_p}(k-1|k-2)$ with $p\in \Ical_i^o(k)$, then $\mathbf{W}_i := [\mathbf{W}^\top_{ip_1} \;\; \mathbf{W}^\top_{ip_2} \cdots ]^\top$ with $\Ical_i^o(k) = \{p_1,p_2,\ldots\}$ is full column rank and $	\Pcal_{\Ical^2_i}(k-1) $ is bounded, i.e., there is $\mathbf{M} \in S_{++}$ such that $\mathbf{M} \succeq \mathbf{P}_{\Ical_i^2}$ for all $\mathbf{P}_{\Ical_i^2} \in \Pcal_{\Ical^2_i}(k-1)$.
\end{lemma}\vspace{-.6cm}
\begin{pf}
	Notice that every row of $\mathbf{W}_{ip}$ is composed of zeros except for one and only one entry, which is unitary. The column where the one is placed is associated with the entry of $\tilde{\mathbf{x}}_{\Ical^2_i}(k-1|k-2)$ that such row extracts. Since $\mathbf{W}_{ip}$ extracts information about $\tilde{\mathbf{x}}_{r}(k-1|k-2)$ for all $r\in \Ical_i^o(k-1)$ and all $p\in \Ical_i^o(k)$, it follows that every column of $\mathbf{W}_i$ has at least one one. As a result, there is a permutation of rows of $\mathbf{W}_i$ that yields a matrix of the form $[\mathbf{I}\;\; \tilde{\mathbf{W}}_i^\top]^\top$, showing that $\mathbf{W}_i$  is full column rank. The fact that $\mathbf{P}_{\Ical_i^2}$ is bounded follows from \citet[Lemma~2]{PedrosoBatistaEtAl2026OCI} since $\mathbf{W}_i$  is full column rank. \hfill$\square$
\end{pf}
\vspace{-.4cm}

\subsection{Overlapping Covariance Intersection Problem}\label{sec:OCI} \vspace{-0.2cm}

One can now incorporate the partial knowledge of $\EV[\mathbf{e}_{\Ical_i}(k)\mathbf{e}_{\Ical_i}^\top(k)]$ analyzed in Section~\ref{sec:partial_info} with the formulation of the design problem as an optimal linear filter in Section~\ref{sec:linear_estimator}. The goal is now to design a gain $\mathbf{K}_{\Ical_i}(k)$ that satisfies $\mathbf{K}_{\Ical_i}(k) \mathbf{H}_{\Ical_i}(k) = \mathbf{I}$ and that minimizes a performance metric on the worst-case bound $\mathbf{X}_{ii}(k|k) \succeq \mathbf{P}_{ii}(k|k)$ under the available information. This can be written as an OCI optimization problem \citep{PedrosoBatistaEtAl2026OCI} \vspace{-0.1cm}
\begin{equation}\label{eq:OCI_orig_prob}
	\begin{aligned}
		&\min_{\substack{\mathbf{K}_{\Ical_i}(k)\in \R^{{n_i}\times o_i}\\\mathbf{X}_{ii}(k|k)\in S^{n_i}_+}}  && J_i(\mathbf{X}_{ii}(k|k))\\
		&\quad\quad\; \mathrm{s.t.} &&  \mathbf{K}_{\Ical_i}(k) \mathbf{H}_{\Ical_i}(k) = \mathbf{I}\\  
		&&&	\mathbf{X}_{ii}(k|k) \succeq  \mathbf{K}_{\Ical_i}(k)(\mathbf{R}_{\Ical_i}(k) + \mathbf{C}_{\Ical_i}(k)\mathbf{P}\mathbf{C}_{\Ical_i}^\top(k)) \mathbf{K}_{\Ical_i}^\top(k),\; \forall \mathbf{P} \in \Pcal_{\Ical^2_i}(k-1).
	\end{aligned}\vspace{-0.1cm}
\end{equation} 
Additionally, the computational unit $\Tcal_i$ must also compute a new bound on $\mathbf{X}_{\Ical_i}(k|k-1)$, which is required for the next iteration of the filter. One can do so resorting to \eqref{eq:P_I_2} and adding a regularization term to \eqref{eq:OCI_orig_prob}, which yields \vspace{-0.1cm}
\begin{equation}\label{eq:OCI_w_reg}
	\begin{aligned}
		&\min_{\substack{\mathbf{K}_{\Ical_i}(k), \mathbf{X}_{ii}(k|k) \succeq \zeros \\ \mathbf{X}_{\Ical_i}(k|k-1) \succeq \zeros}}  && J_i(\mathbf{X}_{ii}(k|k)) + \gamma_i G_i(\mathbf{X}_{\Ical_i}(k|k-1))\\
		&\quad\quad\; \mathrm{s.t.} &&  \mathbf{K}_{\Ical_i}(k) \mathbf{H}_{\Ical_i}(k) = \mathbf{I}\\  
		&&&	\mathbf{X}_{ii}(k|k) \succeq  \mathbf{K}_{\Ical_i}(k)(\mathbf{R}_{\Ical_i}(k) + \mathbf{C}_{\Ical_i}(k)\mathbf{P}\mathbf{C}_{\Ical_i}^\top(k)) \mathbf{K}_{\Ical_i}^\top(k),\; \forall \mathbf{P} \in\Pcal_{\Ical^2_i}(k-1)\\
	&&&	\mathbf{X}_{\Ical_i}(k|k-1) \succeq \mathbf{D}_{\Ical_i}(k-1) \mathbf{P}\mathbf{D}^\top_{\Ical_i}(k-1)+ \mathbf{E}_{\Ical_i}(k-1),  \forall \mathbf{P} \in\Pcal_{\Ical^2_i}(k-1),
	\end{aligned}\vspace{-0.1cm}
\end{equation}
where $\gamma_i >0$ is a small regularization weight and $G_i$ is any performance metric that satisfies Assumption~\ref{ass:J}. It is worth pointing out immediately that \eqref{eq:OCI_w_reg} is not easy or efficient to solve numerically with off-the-shelf solvers. In the following section, a computationally efficient solution to \eqref{eq:OCI_w_reg}  is studied resorting to the literature on the OCI problem \citep{PedrosoBatistaEtAl2026OCI}.

The local design of the gains resorting to \eqref{eq:OCI_w_reg} is distributed and can be computed recursively. We introduce an assumption on the initial information that must be known to solve \eqref{eq:OCI_w_reg} recursively.

\begin{assumption}\label{ass:initial_cond}
	Every computational unit $\Tcal_i$ with $i = 1,2,\ldots,N$ has access to an initial bound $\mathbf{X}_{\Ical_i}(1|0) \succeq \mathbf{P}_{\Ical_i}(1|0)$. \hfill$\triangle$
\end{assumption}

\begin{remark}
	Notice that \eqref{eq:OCI_w_reg} can only be solved starting at $k = 2$. Indeed, at $k = 1$ all the information for a standard filter design is available (e.g., each computational unit has all the information to compute a local minimum variance gain \eqref{eq:optimal_K}). Since any choice for the gain at $k = 1$ allows to recursively employ the design procedure in \eqref{eq:OCI_w_reg} starting at $k =2$, we disregard the details of the synthesis at $k = 1$. \hfill$\triangle$
\end{remark}

\section{Distributed Synthesis Algorithm}\label{sec:synth_alg}

In this section, we study the recursive implementation of the distributed synthesis solution that was cast as an OCI problem in Section~\ref{sec:distributed_synthesis}. In Section~\ref{sec:eff_synth}, we focus on the numerical computation of the local synthesis problem, by providing a solution that can be efficiently solved as a semidefinite program (SDP). In Section~\ref{sec:feasible_recursive}, we establish control-theoretic properties of the recursive implementation of the filter, namely feasibility and consistency. In Section~\ref{sec:alg}, we outline the algorithm of the distributed synthesis procedure from the point of view of a computational unit, showing that the ultra large-scale requirements are met.

\subsection{Efficient Distributed Synthesis}\label{sec:eff_synth} \vspace{-0.2cm}

In the recursive implementation of the distributed filtering approach presented in Section~\ref{sec:distributed_synthesis}, the optimization problem \eqref{eq:OCI_w_reg} must be numerically solved in real-time in the computational unit $\Tcal_i$ for every instant $k \geq 2$. As a result, the viability of the proposed solution depends heavily on the computational efficiency of solving \eqref{eq:OCI_w_reg}. Unfortunately, as already mentioned, the optimization problem \eqref{eq:OCI_w_reg} is nonlinear and cannot be solved resorting to off-the-shelf solvers due to the form of the second and third constraints. 

The OCI problem has been recently studied in \citet{PedrosoBatistaEtAl2026OCI}. Therein, a solution approach to the OCI problem is proposed, which amounts to replacing the information constraints that characterize $\Pcal_{\Ical^2_i}(k-1)$ with a parameterized set of conservative bounds. The new parameterized set of admissible matrices $\mathbf{P}_{\Ical_i^2}(k-1|k-2)$ is called a \emph{family of bounds} and a solution of \eqref{eq:OCI_w_reg} when restricted to that family is called a \emph{family-optimal solution}. In what follows, we characterize a meaningful family of bounds, whose family-optimal solution allows for an efficient numerical computation.

First, it is beneficial to provide an alternative characterization of $\Pcal_{\Ical^2_i}(k-1)$ as shown in the following lemma.

\begin{lemma}\label{lem:inverse_bounds}
	The set $\Pcal_{\Ical^2_i}(k-1)$ of admissible matrices $\mathbf{P}_{\Ical_i^2}(k-1|k-2)$ defined in \eqref{eq:admissible_P} can be equivalently expressed as $\Pcal_{\Ical^2_i}(k-1)  := \{\mathbf{P}_{\Ical_i^2} \succ \zeros : \mathbf{P}_{\Ical_i^2}^{-1} \!\succeq \mathbf{Y}^i_{\Ical_p}(k-1|k-2) \;\forall p \in \Ical_i^o(k)\}$, where $\mathbf{Y}^i_{\Ical_p}(k-1|k-2) := \mathbf{W}_{ip}^\top (k-1) \mathbf{X}_{\Ical_p}(k-1|k-2)^{-1} \mathbf{W}_{ip}(k-1)$ for $p \in \Ical_i^o(k)$.  
\end{lemma}\vspace{-0.5cm}
\begin{pf}
	The result follows from the direct application of \citet[Lemma~1]{PedrosoBatistaEtAl2026OCI}. \hfill$\square$
\end{pf}
\vspace{-0.2cm}

Second, we introduce the so-called \emph{Kahan} family of bounds, which is studied in \citet{Kahan1968} and is a generalization of the most common family of bounding ellipsoids used in standard CI problems (e.g., \citet{ReinhardtNoackEtAl2015}). Specifically, the Kahan family of bounds is denoted by $\Pcal^{\mathrm{KF}}_{\Ical^2_i}(k-1)$ and it is parameterized by a vector $\boldsymbol{\omega} \in  \Delta^{|\Ical_i^o(k)|} := \{\boldsymbol{\omega} \in \Rnn^{|\Ical_i^o(k)|} : \ones^\top \boldsymbol{\omega} = 1 \}$ as
\begin{equation}\label{eq:admissible_P_KF}
	\Pcal^{\mathrm{KF}}_{\Ical^2_i}(\bomega,k-1) = \left\{\mathbf{P}_{\Ical_i^2} \succ \zeros : \mathbf{P}_{\Ical_i^2}^{-1} \!\succeq \sum\nolimits_{p\in \Ical_i^o(k)} \bomega_p\mathbf{Y}^i_{\Ical_p}(k-1|k-2)\right\},
\end{equation}
where, for the sake of conciseness and by abuse of notation, each entry of $\bomega$ is associated with one element of $\Ical_i^o(k)$ and the entry associated with $p \in \Ical_i^o(k)$ is denoted by $\bomega_p$. The following lemma establishes that $\Pcal^{\mathrm{KF}}_{\Ical^2_i}(\bomega,k-1)$ is a conservative approximation of $\Pcal_{\Ical^2_i}(k-1)$. For more details on the conservativeness and properties of the Kahan family see \citet{Kahan1968} and \citet[Section~III.C]{PedrosoBatistaEtAl2026OCI}.

\begin{lemma}\label{lem:KF_conservative}
	The Kahan family is conservative, i.e., $\Pcal_{\Ical^2_i}(k-1) \subseteq \Pcal^{\mathrm{KF}}_{\Ical^2_i}(\bomega,k-1)$ for all  $\boldsymbol{\omega} \in  \Delta^{|\Ical_i^o(k)|}$.
\end{lemma}\vspace{-1cm}
\begin{pf}
	The result amounts to showing that $\mathbf{P} \in \Pcal_{\Ical^2_i}(k-1)$ implies $\mathbf{P}\in \Pcal^{\mathrm{KF}}_{\Ical^2_i}(\bomega,k-1)$ for all $\boldsymbol{\omega} \in  \Delta^{|\Ical_i^o(k)|}$. Take any $\mathbf{P} \in \Pcal_{\Ical^2_i}(k-1)$. By Lemma~\ref{lem:inverse_bounds}, $\mathbf{P}^{-1} \succeq \mathbf{Y}^i_{\Ical_p}(k-1|k-2)$ for all $p \in \Ical_i^o(k)$. As a result, for all $\boldsymbol{\omega} \in  \Delta^{|\Ical_i^o(k)|}$ one can write $\mathbf{P}^{-1} = \sum\nolimits_{p\in \Ical_i^o(k)} \bomega_p\mathbf{P}^{-1}  \succeq \sum\nolimits_{p\in \Ical_i^o(k)} \bomega_p\mathbf{Y}^i_{\Ical_p}(k-1|k-2)$. Therefore, $\mathbf{P} \in \Pcal^{\mathrm{KF}}_{\Ical^2_i}(\bomega,k-1)$ according to the definition in  \eqref{eq:admissible_P_KF}. \hfill$\square$
\end{pf}
\vspace{-0.4cm}

The Kahan-family-optimal solution of the design problem of computational unit $\Tcal_i$ at time $k$ in \eqref{eq:OCI_w_reg} is then said to be the solution to
\begin{equation}\label{eq:OCI_KF}
	\begin{aligned}
		&\min_{\substack{\mathbf{K}_{\Ical_i}(k), \mathbf{X}_{ii}(k|k) \succeq \zeros\\ \mathbf{X}_{\Ical_i}(k|k-1) \succeq \zeros, \boldsymbol{\omega} \in  \Delta^{|\Ical_i^o(k)|}}}  && J_i(\mathbf{X}_{ii}(k|k)) + \gamma_i G_i(\mathbf{X}_{\Ical_i}(k|k-1))\\
		&\quad\quad \quad\quad\;\;\,\mathrm{s.t.} &&  \mathbf{K}_{\Ical_i}(k) \mathbf{H}_{\Ical_i}(k) = \mathbf{I}\\  
		&&&	\mathbf{X}_{ii}(k|k) \succeq  \mathbf{K}_{\Ical_i}(k)(\mathbf{R}_{\Ical_i}(k) + \mathbf{C}_{\Ical_i}(k)\mathbf{P}\mathbf{C}_{\Ical_i}^\top(k)) \mathbf{K}_{\Ical_i}^\top(k),\; \forall \mathbf{P} \in 	\Pcal^{\mathrm{KF}}_{\Ical^2_i}(\bomega,k-1)\\
		&&&	\mathbf{X}_{\Ical_i}(k|k-1) \succeq \mathbf{D}_{\Ical_i}(k-1) \mathbf{P}\mathbf{D}^\top_{\Ical_i}(k-1)+ \mathbf{E}_{\Ical_i}(k-1),   \forall \mathbf{P}	\in \Pcal^{\mathrm{KF}}_{\Ical^2_i}(\bomega,k-1)
	\end{aligned}
\end{equation}
where $|\Ical_i^o(k)|-1$ more degrees of freedom are added w.r.t.\ \eqref{eq:OCI_w_reg} to parameterize the Kahan family.
Crucially, using a result from \citet{PedrosoBatistaEtAl2026OCI}, one can show that the Kahan-family-optimal solution of  \eqref{eq:OCI_w_reg} can be efficiently computed as an SDP, as shown in the following result.

\begin{thm}\label{th:KF_sol}
	Under Assumptions~\ref{ass:J} and~\ref{ass:Q_pd}, the triple $(\mathbf{K}_{\Ical_i}(k), \mathbf{X}_{ii}(k|k), \mathbf{X}_{\Ical_i}(k|k-1))$ is a Kahan-family-optimal solution to the design problem of computational unit $\Tcal_i$ at time $k$ in \eqref{eq:OCI_w_reg}, where $(\mathbf{U},\mathbf{Y},\mathbf{X}_{ii}(k|k),\mathbf{X}_{\Ical_i}(k|k-1), \boldsymbol{\omega})\in$
	\begin{equation}\label{eq:KF_SDP}
		  \begin{aligned}
		 	 \!\!\!\!\!\!\!\!\!\!& \underset{\substack{\mathbf{U}  \succeq \zeros,\mathbf{Y}  \succeq \zeros,  \mathbf{X}_{ii}(k|k) \succeq \zeros \\\mathbf{X}_{\Ical_i}(k|k-1) \succeq \zeros, \boldsymbol{\omega} \in  \Delta^{|\Ical_i^o(k)|}}}{\argmin} &\quad \!\!\!\!\!\!\!\!\!\!\!\!\!\!\!\!\!\!\!\!\!\!& \!\!\!\!\!\! J_i(\mathbf{X}_{ii}(k|k)) + \gamma_i G_i(\mathbf{X}_{\Ical_i}(k|k-1))\\
		 	& \quad\quad \qquad\;\,\mathrm{s.t.} && \!\!\!\!\!\!\!\!\begin{bmatrix} \mathbf{X}_{ii}(k|k) &  \mathbf{I} \\ \mathbf{I} & \mathbf{H}^\top_{\Ical_i}(k) \mathbf{R}_{\Ical_i}(k)^{-1}\mathbf{H}_{\Ical_i}(k)-\mathbf{U}\end{bmatrix}  \succeq \zeros\\
		 	&  &&\!\!\!\!\!\! \!\!\begin{bmatrix} \mathbf{U}&  \mathbf{H}^\top_{\Ical_i}(k)\mathbf{R}_{\Ical_i}(k)^{-1}\mathbf{C}_{\Ical_i}(k)\\ (\mathbf{H}^\top_{\Ical_i}(k)\mathbf{R}_{\Ical_i}(k)^{-1}\mathbf{C}_{\Ical_i}(k) )^\top & \mathbf{Y} +\mathbf{C}^\top_{\Ical_i}(k) \mathbf{R}_{\Ical_i}(k)^{-1}\mathbf{C}_{\Ical_i}(k)\end{bmatrix} \! \succeq \!\zeros\\
		 	&&&\!\!\!\!\!\! \!\!\begin{bmatrix} \mathbf{X}_{\Ical_i}(k|k-1)-\mathbf{E}_{\Ical_i}(k-1) &  \mathbf{D}_{\Ical_i}(k-1)\\ \mathbf{D}^\top_{\Ical_i}(k-1) & \mathbf{Y}\end{bmatrix}  \succeq \zeros\\
		 	&&&\!\!\!\!\!\! \mathbf{Y} =  \sum\nolimits_{p\in \Ical_i^o(k)} \bomega_p\mathbf{Y}^i_{\Ical_p}(k-1|k-2)
		 \end{aligned} \vspace{-0.4cm}
	\end{equation}
and  \vspace{-0.4cm}
	\begin{equation}\label{eq:KF_K}
	\begin{split}
	\mathbf{K}_{\Ical_i}(k) := &\left(\mathbf{H}^\top_{\Ical_i}(k)\mathbf{R}^{-1}_{\Ical_i}(k)\left(\mathbf{R}_{\Ical_i}(k)-\mathbf{C}_{\Ical_i}(k)\left( \mathbf{Y}+\mathbf{C}^\top_{\Ical_i}(k) \mathbf{R}^{-1}_{\Ical_i}(k)\mathbf{C}_{\Ical_i}(k)\right)^{+} \mathbf{C}^\top_{\Ical_i}(k)\right)\mathbf{R}^{-1}_{\Ical_i}(k)\mathbf{H}_{\Ical_i}(k)\right)^{-1} \\
	& \quad \! \mathbf{H}^\top_{\Ical_i}(k)\mathbf{R}^{-1}_{\Ical_i}(k)\left(\mathbf{R}_{\Ical_i}(k)-\mathbf{C}_{\Ical_i}(k)\left( \mathbf{Y}+\mathbf{C}^\top_{\Ical_i}(k) \mathbf{R}^{-1}_{\Ical_i}(k)\mathbf{C}_{\Ical_i}(k)\right)^{+} \mathbf{C}^\top_{\Ical_i}(k)\right)\mathbf{R}^{-1}_{\Ical_i}(k).
	\end{split}
\end{equation}
\end{thm}\vspace{-0.8cm}
\begin{pf}
	Under Assumption~\ref{ass:J} and since $\mathbf{R}_{\Ical_i}(k) \succ \zeros$, by Lemma~\ref{lem:R_pd} under Assumption~\ref{ass:Q_pd}, we are in the conditions of \citet[Theorem~1 and Corollary~2]{PedrosoBatistaEtAl2026OCI}, which immediately proves the result. \hfill$\square$
\end{pf}

Notice that the feasible set of \eqref{eq:KF_SDP} is convex, since it is characterized by three linear matrix inequalities and two equality constraints. Thus, if $J_i$ and $G_i$ are convex (which is the case for meaningful choices such as the trace or determinant\footnote{If one desires to use the determinant as the metric $J_i$, then the optimization problem \eqref{eq:KF_SDP} needs to be slightly modified to be cast as a SDP (see \citet[Remark~III.3]{PedrosoBatistaEtAl2026OCI} and \citet[Section~6.2.3]{MOSEK2024} for details).}), then \eqref{eq:KF_SDP} is a convex optimization problem, which has desirable properties such as robustness to changes in input parameters and the existence of efficient numerical algorithms with global optimality guarantees \citep{BoydVandenberghe2004}. Furthermore, for meaningful choices of $J_i$ and $G_i$ (such as the trace or determinant\footnotemark[\value{footnote}]), it is possible to write \eqref{eq:KF_SDP} as an SDP, for which well-performing off-the-shelf solvers with polynomial worst-case complexity exist \citep{VandenbergheBoyd1996, MOSEK2024}.

\vspace{0.3cm}
\begin{mdframed}[style=callout]
	From Theorem~\ref{th:KF_sol}, one concludes that it is possible to compute a Kahan-family-optimal solution to the design problem of computational unit $\Tcal_i$ at time $k$ from the numerical solution of \eqref{eq:KF_SDP}. Moreover, for meaningful choices of $J_i$ and $G_i$ such as the trace (which corresponds to minimum variance estimation) or determinant, \eqref{eq:KF_SDP} can be cast as an SDP, which can be efficiently solved with well-performing off-the-shelf solvers.
\end{mdframed}

\subsection{Recursive Feasibility and Consistency}\label{sec:feasible_recursive} \vspace{-0.2cm}

Crucially, the distributed filtering solution that results from designing the Kahan-family-optimal gains recursively according to Theorem~\ref{th:KF_sol} in every agent has very strong properties, namely \emph{recursive feasibility} and \emph{recursive consistency}, given an initial condition according to Assumption~\ref{ass:initial_cond}.
Recursive feasibility guarantees that, for all $k\geq 2$ and all $i = 1,2,\ldots,N$, the optimization problem \eqref{eq:KF_SDP} is feasible, i.e., there exists a tuple of decision variables that satisfies the constraints. Recursive consistency guarantees that, for all $k\geq 2$ and all $i = 1,2,\ldots,N$, the estimates $\hat{\mathbf{x}}_i(k|k)$ are consistent according to Definition~\ref{def:consistency}, i.e., $\mathbf{X}_{ii}(k|k)\succeq \EV[\tilde{\mathbf{x}}_i(k|k)\tilde{\mathbf{x}}_i^\top(k|k)]$. The following theorem is the main result of the paper.

\begin{thm}\label{th:feasible_consistent}
	Under Assumptions~\ref{ass:J}--~\ref{ass:initial_cond}, the distributed CL filter \eqref{eq:dist_filter} with Kahan-family-optimal gains computed recursively resorting to Theorem~\ref{th:KF_sol} is recursively feasible and recursively consistent. 
\end{thm} \vspace{-0.8cm}
\begin{pf}
First, from Lemma~\ref{lem:W_full_col_rank}, for all $i =1, 2, \ldots, N$, if at time $k$ there exist bounds $\mathbf{X}_{\Ical_p}(k-1|k-2) \succeq \mathbf{P}_{\Ical_p}(k-1|k-2)$ with $p\in \Ical_i^o(k)$, then $\mathbf{W}_i$ is full column rank and $\Pcal_{\Ical^2_i}(k-1) $ is bounded. Thus, it follows from \cite[Corollary~1]{PedrosoBatistaEtAl2026OCI} that \eqref{eq:OCI_w_reg} is feasible and from \cite[Corollary~2]{PedrosoBatistaEtAl2026OCI} that \eqref{eq:KF_SDP} is feasible. Since at time $k$, computational unit $\Tcal_i$ computes a new $\mathbf{X}_{\Ical_i}(k|k-1)$ resorting to \eqref{eq:KF_SDP}, then it follows from a simple induction argument that the filtering solution is recursively feasible. Second, notice that recursive consistency follows immediately from the construction of \eqref{eq:OCI_w_reg}, which is preserved in  \eqref{eq:OCI_KF} since the Kahan family is conservative by Lemma~\ref{lem:KF_conservative}. Since \eqref{eq:KF_SDP} characterizes a solution of \eqref{eq:OCI_KF}, it follows that the filtering solution is recursively consistent.  \hfill$\square$
\end{pf}

\begin{mdframed}[style=callout,nobreak=true]
	As pointed out in the state-of-art in Section~\ref{sec:introduction} as well as in \cite[Section~5]{PedrosoBatistaEtAl2025ULSS}, there is a gap to develop well-performing CL algorithms that are both consistent and scalable to an ultra large scale. Theorem~\ref{th:feasible_consistent} shows that the proposed distributed CL that relies on a distributed synthesis procedure resorting to local information is provably consistent.
\end{mdframed}

\subsection{Synthesis Algorithm}\label{sec:alg} \vspace{-0.2cm}

Fig.~\ref{fig:alg} depicts the computation and communication steps of the proposed CL algorithm from the point of view of computational unit $\Tcal_i$. Different block colors represent different timing requirements. First, red blocks correspond to taking measurements, updating the state estimate, and updating the control action. These steps must be performed across all computational units precisely at times $t = kT$ with $k\in \N$, where $T$ is the sampling time of the discrete-time system \eqref{eq:sys_model}. Second, green blocks correspond to local computations, which do not have any synchronization requirements. Third, orange blocks correspond to the exchange of information between neighboring computational units. These transmissions need not occur at a precise time, but they must naturally reach neighboring computational units on time. 

\vspace{0.3cm}
\begin{mdframed}[style=callout]
	Notice that, as long as the sampling time $T$ is greater than the sum of computational and communication delays, the proposed CL algorithm is robust to those delays. In conclusion, together with the fact that computation, memory, and communication load of each operation in Fig.~\ref{fig:alg} does not scale with $N$, the proposed CL algorithm satisfies the ultra large-scale feasibility requirements \cite[Table~1]{PedrosoBatistaEtAl2025ULSS}.
\end{mdframed}

\begin{figure*}[t]
	\raggedright
	\includegraphics{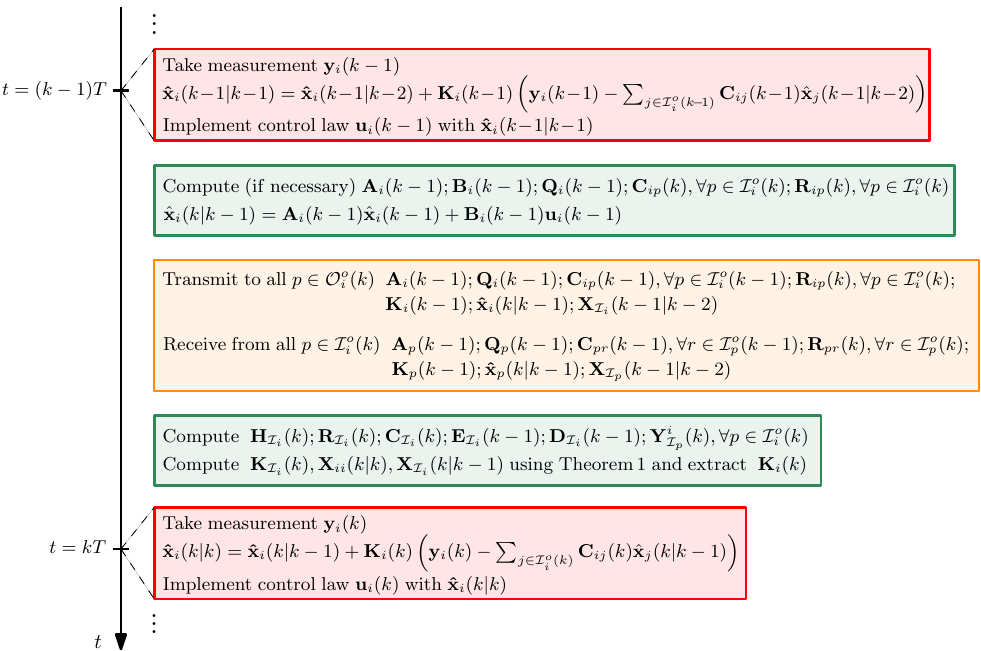}
	\caption{Timeline of proposed CL algorithm for a single computational unit $\Tcal_i$.}
	\label{fig:alg}
\end{figure*}


\section{Numerical  Results}\label{sec:results}

In this section, we simulate the proposed CL algorithm on illustrative synthetic systems.  We consider a grid topology of a network of $N$ systems with discrete-time single integrator dynamics. Each agent has access to relative position measurements w.r.t.\ neighboring systems according to the grid output topology, which is depicted in Fig.~\ref{fig:gridN36} for the case $N = 36$. Some systems have access to measurements of their absolute positions, namely systems $\Scal_1, \Scal_4, \Scal_{28}$, and $\Scal_{35}$ in the network in Fig.~\ref{fig:gridN36}.

\begin{figure}[ht]
	\centering
	\includegraphics[width=0.35\textwidth]{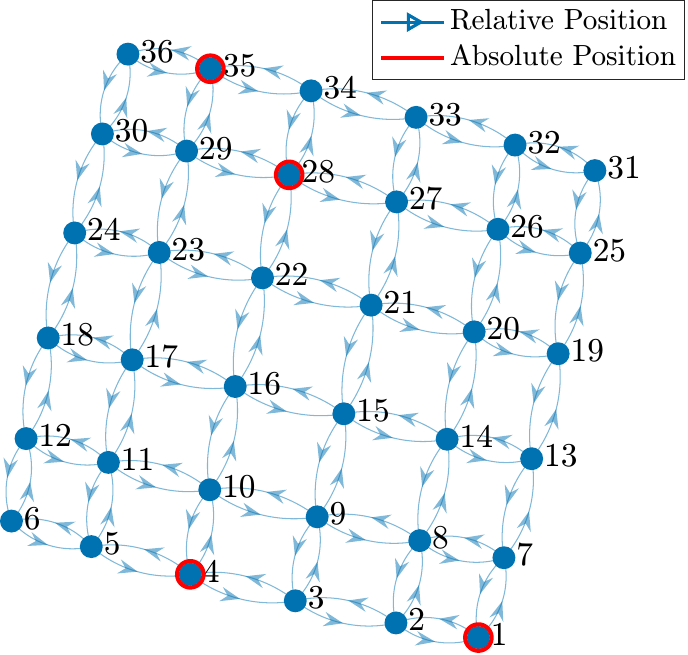}
	\caption{Output topology of the illustrative network for $N = 36$.}
	\label{fig:gridN36}
\end{figure}

\begin{figure}[b]
	\centering
	\begin{subfigure}[t]{0.45\textwidth}
		\includegraphics[width=\textwidth]{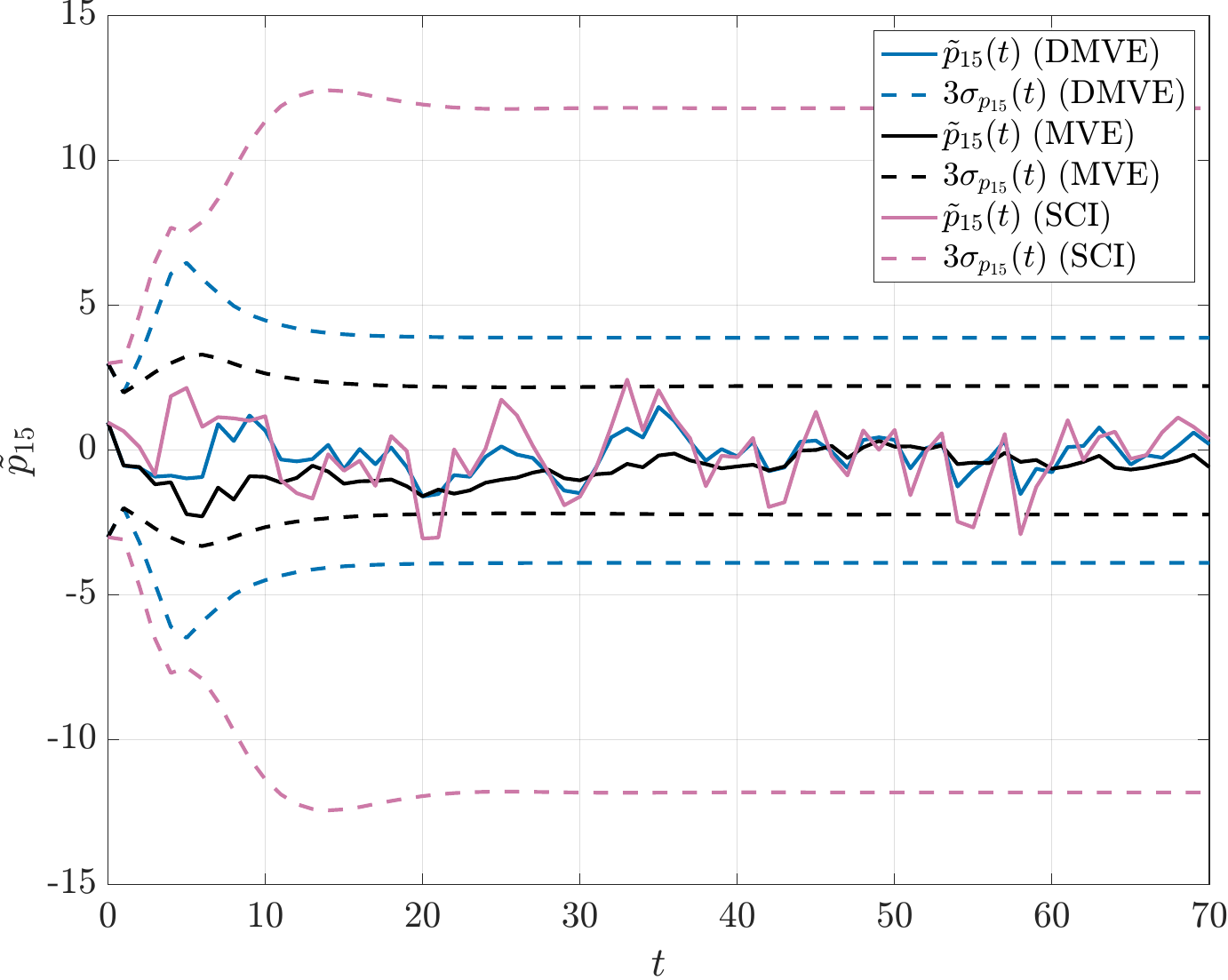}
		\vspace{-0.5cm}
		\caption{Position estimation error of $\Scal_{15}$.}
	\end{subfigure}\hspace{0.5cm}%
	\begin{subfigure}[t]{0.45\textwidth}
		\includegraphics[width=\textwidth]{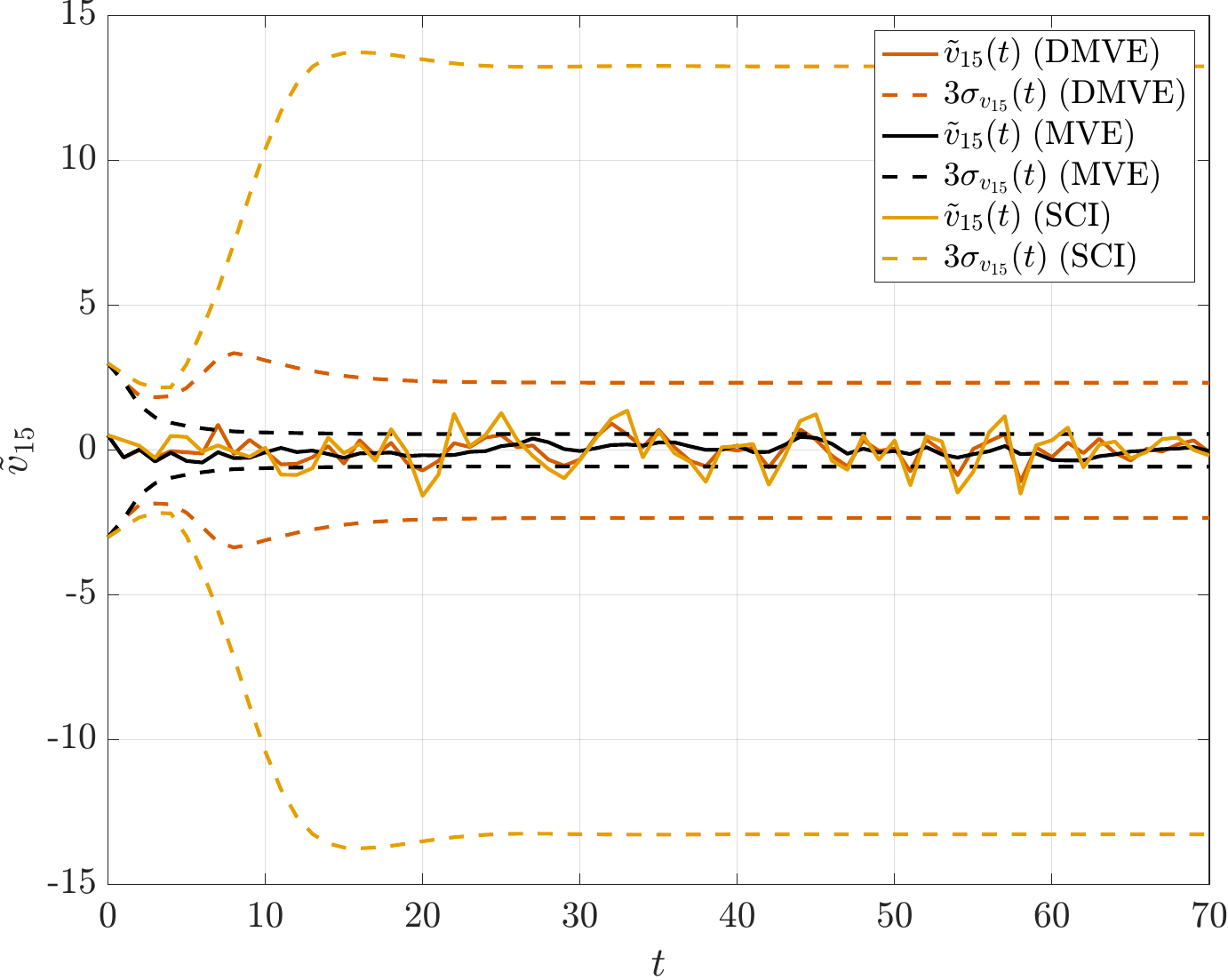}
		\vspace{-0.5cm}
		\caption{Velocity estimation error of $\Scal_{15}$.}
	\end{subfigure}
	\centering
	\begin{subfigure}[t]{0.45\textwidth}
		\includegraphics[width=\textwidth]{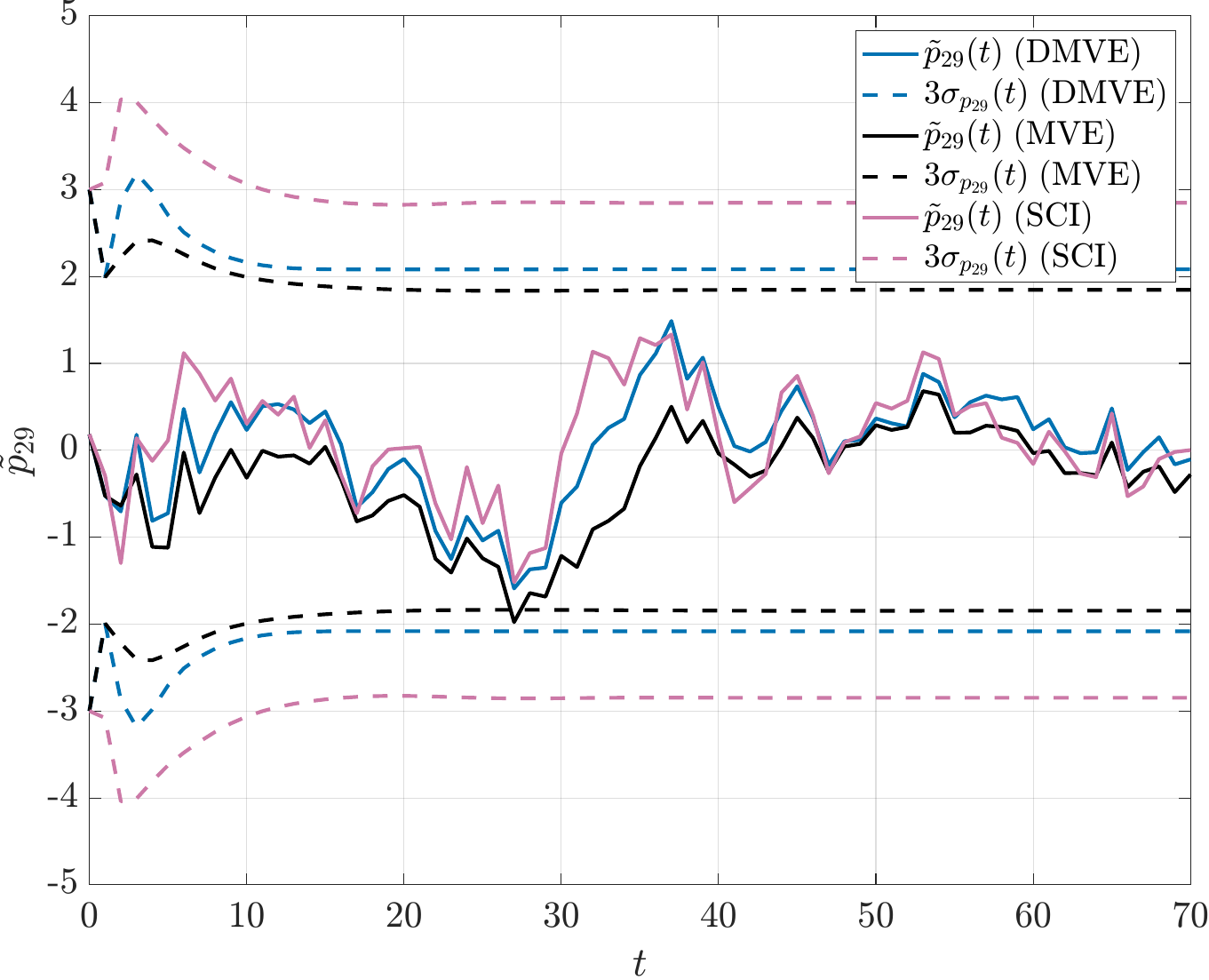}
		\caption{Position estimation error of $\Scal_{29}$.}
		\vspace{-0.5cm}
	\end{subfigure}\hspace{0.5cm}%
	\begin{subfigure}[t]{0.45\textwidth}
		\includegraphics[width=\textwidth]{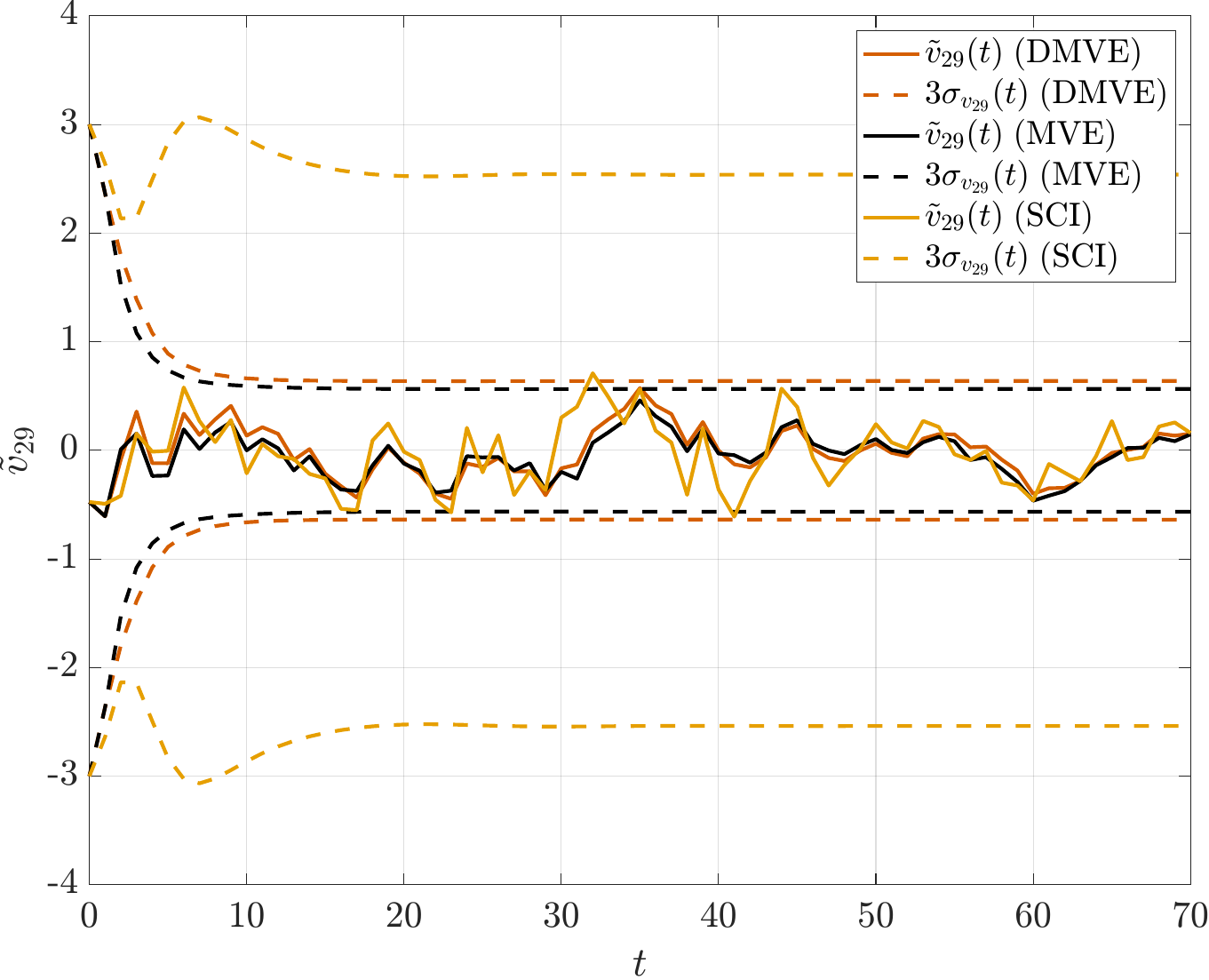}
		\vspace{-0.5cm}
		\caption{Velocity estimation error of $\Scal_{29}$.}
	\end{subfigure}
	\caption{Evolution of estimation error of $\Scal_{15}$ and $\Scal_{29}$ and respective three-standard-deviation bounds.}
	\label{fig:error_gridN36}
\end{figure}

Fig.~\ref{fig:error_gridN36} depicts the evolution of the estimation error of two systems of the network depicted in Fig.~\ref{fig:gridN36} according to the proposed cooperative estimation algorithm with a minimum variance design criterion, which we label distributed minimum variance estimation (DMVE). Fig.~\ref{fig:error_gridN36} also depicts $3\sigma$ bounds according to the consistent estimation error covariance bounds provided by the algorithm. Fig.~\ref{fig:error_gridN36} also depicts the evolution of the estimation error and $3\sigma$ bounds for: (i)~a distributed minimum variance algorithm that is designed centrally, and thus has exact global covariance propagation, as described in Section~\ref{sec:MVE} (which we label MVE); and (ii)~the SCI algorithm proposed in \cite{LiNashashibiEtAl2013}. To evaluate the performance of the three methods we use the accuracy metric $\mathrm{D}$ defined as
\begin{equation*}
	\mathrm{D}:= \frac{1}{N (T_{\text{sim}}+1)}\sum_{k= 0}^{T_{\text{sim}}} \sum_{i = 1}^N ||\tilde{\mathbf{x}}_i(k|k)||_2, 
\end{equation*}
and the covariance tightness metric $\mathrm{ANEES}$ defined as
\begin{equation*} 	\mathrm{ANEES}:= \frac{1}{N  (T_{\text{sim}}+1)}\sum_{k= 0}^{T_{\text{sim}}} \sum_{i = 1}^N \tilde{\mathbf{x}}_i^\top(k|k)\mathbf{X}_{ii}^{-1}(k|k)\tilde{\mathbf{x}}_i(k|k),
\end{equation*}
which are similar to the metrics used, for instance, in \cite{LuftEtAl2018}. The accuracy metric $\mathrm{D}$ quantifies the average norm of the state estimation error across the entire network, thus, a lower $\mathrm{D}$ indicates higher estimation accuracy. The covariance tightness metric $\mathrm{ANEES}$ quantifies the average normalized estimation error squared. Thus, for a provably consistent estimation method, a higher $\mathrm{ANEES}$ indicates lower conservatism of the consistent covariance bounds $\mathbf{X}_{ii}$ for $i = 1,2,\ldots,N$. The performance metrics averaged over $100$ simulations with $T_{\text{sim}} = 70$ of the illustrative network in Fig.~\ref{fig:gridN36} are shown in Table~\ref{tab:metrics_N36}. The estimation accuracy of the proposed CL method (DMVE) is not significantly worse (only about $8\%$ worse in the $\mathrm{D}$ metric) than the minimum variance estimator with exact global covariance propagation (MVE). However, we notice a significant conservatism introduced by DMVE with respect to MVE of about $70\%$ of the $\mathrm{ANEES}$ metric.  In Table~\ref{tab:metrics_N576} we also show the  performance metrics averaged over $100$ simulations with $T_{\text{sim}} = 70$ of a bigger network with $N = 576$.

\begin{table}[t]
	\vspace{0.4cm}
	\centering
	\renewcommand{\arraystretch}{1.2}
	\caption{Accuracy ($\mathrm{D}$) and covariance tightness ($\mathrm{ANEES}$) metrics averaged over $100$ simulations with $T_{\text{sim}} = 70$ of the illustrative network in Fig.~\ref{fig:gridN36} with $N = 36$.\vspace{0.4cm}}
	\label{tab:metrics_N36}
	\begin{tabular}{lcccc}
		\hline
		Method & $\mathrm{D}$ & &  $\mathrm{ANEES}$ & \\
		\hline
		$\mathrm{\mathbf{DMVE}}$ & $\mathbf{0.7579}$ & $-$       & $\mathbf{1.245}$ & $-$ \\
		$\mathrm{MVE}$  & $0.6995$ & \textcolor{darkgreen}{$-7.706\%$} & $2.076$ &  \textcolor{darkgreen}{$+66.75\%$} \\
		$\mathrm{SCI}$  & $0.9724$ & \textcolor{darkred}{$+28.30\%$} & $0.6734$ & \textcolor{darkred}{$-45.93\%$} \\
		\hline
	\end{tabular}
	\vspace{0.4cm}
\end{table}

\begin{table}[t]
	\vspace{0.4cm}
	\centering 
	\renewcommand{\arraystretch}{1.2}
	\caption{Accuracy ($\mathrm{D}$) and covariance tightness ($\mathrm{ANEES}$) metrics averaged over $100$ simulations with $T_{\text{sim}} = 70$ with $N = 576$.\vspace{0.4cm}}
	\label{tab:metrics_N576}
	\begin{tabular}{lcccc}
		\hline
		Method & $\mathrm{D}$ & &  $\mathrm{ANEES}$ & \\
		\hline
		$\mathrm{\mathbf{DMVE}}$ & $0.7175$ & --       & $1.126$ & --\\
		$\mathrm{MVE}$  & $0.5844$ &  \textcolor{darkgreen}{$-18.54\%$} & $1.973$ &  \textcolor{darkgreen}{$+75.24\%$}\\
		$\mathrm{SCI}$  & $0.9594$ &  \textcolor{darkred}{$+33.71\%$} & $0.6496$ &  \textcolor{darkred}{$-42.29\%$}\\
		\hline
	\end{tabular}
	\vspace{0.4cm}
\end{table}

\vspace{0.3cm}
\begin{mdframed}[style=callout]
	The comparison between the two distributed CL methods DMVE and SCI, which are the only ones scalable to an ultra large scale, is the most insightful. Indeed, from Table~\ref{tab:metrics_N36}, the proposed CL algorithm DMVE achieves significantly better estimation accuracy (roughly a $30\%$ improvement in the $\mathrm{D}$ metric) and it is significantly less conservative (roughly a $40\%$ improvement in the $\mathrm{ANEES}$ metric). The accuracy and conservatism improvements are also clearly visible in Fig.~\ref{fig:error_gridN36}.
\end{mdframed}

Another interesting analysis is the synthesis time of a single filter iteration of each of the CL methods. The three methods were simulated for many network sizes $N$ and the average synthesis time of each iteration was measured. It is important to remark that the proposed CL method (DMVE) and SCI were synthesized in parallel for each agent, whereas that is not possible for the centralized design of MVE. The synthesis times are depicted in Fig.~\ref{fig:runtime} in a logarithmic scale. It is worth remarking that the absolute synthesis times depicted in Fig.~\ref{fig:runtime} are not very insightful, since the implementation is rather naive and not optimized for performance (e.g. the SDPs for DMVE are parsed at every iteration). Insights can be instead gathered from their qualitative growth.

\begin{figure}[ht]
	\centering
	\includegraphics[width=0.5\textwidth]{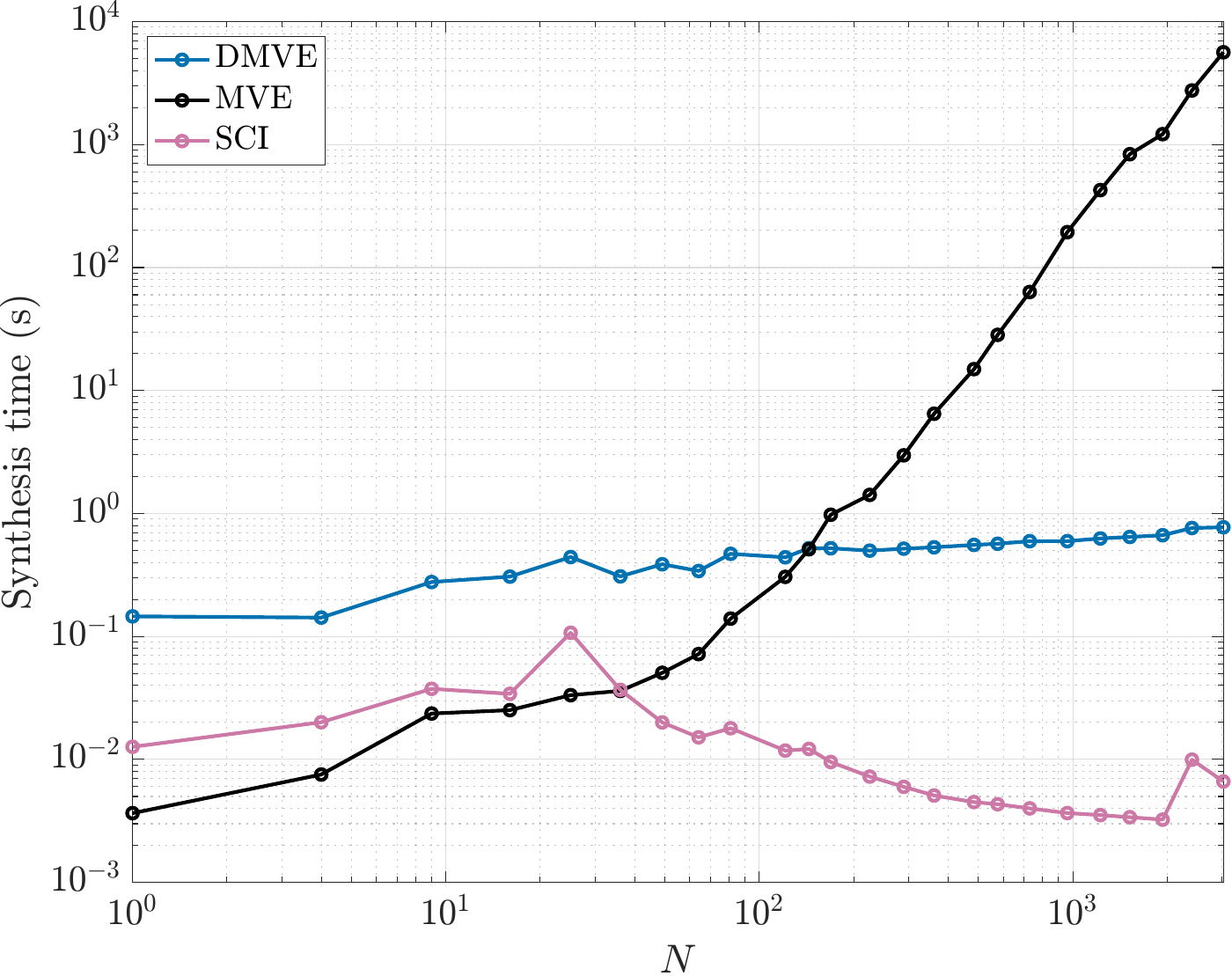}
	\caption{Evolution with $N$ of the synthesis time of a single iteration of DMVE, MVE, and SCI.}
	\label{fig:runtime}
\end{figure}

\vspace{0.3cm}
\begin{mdframed}[style=callout]
	One can readily notice that while the synthesis times of DMVE and SCI (which are the only ones scalable to an ultra large scale) remain approximately constant as $N$ grows, that is not the case for MVE. Indeed, the asymptotic slope of the synthesis time of MVE in the logarithmic scale is approximately $3$, which matches the expected growth with $N^3$ that is associated with the inversion and multiplication of matrices whose dimension grows with $N$. 
\end{mdframed}

An implementation of the solution of the proposed CL method as well as the code of all numerical simulations in this paper is available in an open-access repository at {\scriptsize \href{https://github.com/decenter2021/coop-loc-oci}{\texttt{github.com/decenter2021/coop-loc-oci}}}.


\section{Conclusion}
This work addresses a long-standing gap in cooperative localization (CL) research: the lack of algorithms that are well-performing, consistent, and scalable to an ultra large scale (ULS). To address this problem, we use the concept of \emph{robustness to information uncertainty}. Specifically, local gains are designed to minimize the worst-case consistent bound on the covariance of the estimation error given the information that is known. First, we conclude that this design problem can be expressed in an \emph{overlapping covariance intersection} framework, which is a CI-like framework with partial structural information that was recently introduced. Second, the proposed algorithm is provably consistent, robust to communication and computation delays, and satisfies the ULS feasibility requirements, since its computational, memory, and communication costs do not scale with the number of agents in the system. Third, numerical simulations demonstrate significant improvements over state-of-the-art ULS-scalable methods. Specifically, the proposed method achieves significantly higher accuracy and lower conservatism compared to split covariance intersection (SCI) while maintaining scalability. These results establish the proposed approach as a strong candidate for practical, consistent, and efficient CL in ultra large-scale systems.


\begin{ack} 
This work was supported by the LARSyS FCT funding (UID/50009/2025:  DOI \href{https://doi.org/10.54499/UID/50009/2025}{10.54499/UID/50009/2025}; LA/P/0083/2020: DOI \href{https://doi.org/10.54499/LA/P/0083/2020}{{10.54499/LA/P/0083/2020}}).%
\end{ack}

\begingroup
\scriptsize
\bibliographystyle{plainnat}
\bibliography{references-gt.bib,bibliography.bib,references-c.bib}
\endgroup

\end{document}